\documentclass[letterpaper]{article}
\usepackage[totalwidth=420pt,totalheight=625pt]{geometry}

\pdfoutput=1

\title{Tractable Answer-Set Programming with Weight Constraints: Bounded Treewidth is not Enough
      \footnote{To appear in Theory and Practice of Logic Programming (TPLP). A preliminary version appeared in the Proceedings of the Twelfth International Conference on Principles of Knowledge Representation and Reasoning (KR 2010).}}
\author{Reinhard Pichler\thanks{Supported by the Austrian Science Fund (FWF): P20704-N18.}, Stefan R\"{u}mmele\footnotemark[2], Stefan Szeider\thanks{Supported by the European Research Council (ERC), project 239962.}, Stefan Woltran\thanks{Supported by Vienna University of Technology
special fund ``Innovative Projekte 9006.09/008''.}}
\date{%
Vienna University of Technology, Austria\\
\texttt{{\rmfamily \{}pichler, ruemmele, woltran{\rmfamily \}}@dbai.tuwien.ac.at, stefan@szeider.net}
}

\usepackage{amsmath}
\usepackage{amssymb}
\usepackage{tikz}
\usepackage{url}
\usepackage{amsthm}

\newtheorem {theorem}               {Theorem}

\newtheorem {definition}  [theorem] {Definition}
\newtheorem {example}     [theorem] {Example}
\newtheorem {lemma}       [theorem] {Lemma}
\newtheorem {proposition} [theorem] {Proposition}

\newcommand{\claim}[2] {\emph{Claim #1: #2}}

\newcommand {\cC} {{\mathcal C}}

\newcommand {\cG} {{\mathcal G}}
\newcommand {\cO} {{\mathcal O}}

\newcommand {\cT} {{\mathcal T}}

\newcommand {\exend}{\ifmmode\hbox{$\dashv$}\else{\unskip\nobreak\hfil
                       \penalty50\hskip1em\null\nobreak\hfil\hbox{$\dashv$}
                       \parfillskip=0pt\finalhyphendemerits=0\endgraf}
                     \fi}

\renewcommand {\mid} {\,:\,}

\newcommand {\coloneqq} {\mathrel{\mathop:}=}

\newcommand {\LR} {\Leftrightarrow}
\newcommand {\Ra} {\Rightarrow}

\newcommand {\ra} {\rightarrow}

\newcommand {\N} {\mathbb{N}}

\newcommand {\card}[1]  {\left| #1 \right|}
\newcommand {\size}[1]  {\lVert #1 \rVert}
\newcommand {\order}[1] {[\![ #1 ]\!]}

\renewcommand {\theta} {\vartheta}
\renewcommand {\phi}   {\varphi}

\newcommand {\bigO} {\cO}

\newcommand {\univ}   {U} %

\newcommand {\atoms}  {A}
\newcommand {\car}    {\Gamma}
\newcommand {\clause} {\mathit{Cl}}
\newcommand {\const}  {\mathcal{C}} %
\newcommand {\low}    {l}
\newcommand {\rules}  {\mathcal{R}}
\newcommand {\upper}  {u}
\newcommand {\weight} {W}

\newcommand {\body}    {B(r)}
\newcommand {\head}    {H(r)}

\newcommand {\AS}[1] {\mathcal{AS}(#1)}
\newcommand {\prog}  {\Pi}

\newcommand {\graph}  {\cG}
\newcommand {\TD}     {\cT}
\newcommand {\tree}   {T}

\newcommand {\bags}   {\chi}
\newcommand {\bag}    {\chi}
\newcommand {\tw}     {\mathit{tw}}

\newcommand {\ww}     {\mathit{cw}}

\newcommand {\skr}     {\mathit{sk}}
\newcommand {\crr}     {\mathit{cr}}

\newcommand {\innode}[2]    {{#1|_{#2}}}
\newcommand {\intree}[2]    {{#1\!\!\downarrow_{#2}}}
\newcommand {\sintree}[2]   {{#1\!\downarrow_{#2}}}
\newcommand {\insubtree}[2] {{#1\!\!\Downarrow_{#2}}}

\newcommand {\rootnode} {n_{root}}

\newcommand {\exprog}  {\Pi_{Ex}}

\newcommand {\assign}    {\vartheta}         %
\newcommand {\pmodel}    {M}                 %
\newcommand {\pconst}    {C}                 %
\newcommand {\prules}    {R}                 %
\newcommand {\porder}    {L_<}               %
\newcommand {\lweight}   {\gamma_<}          %
\newcommand {\pweight}    {\gamma}           %

\newcommand {\derivation} {\Delta}           %
\newcommand {\drules}     {\delta_\prules}   %
\newcommand {\datoms}     {\delta_\pmodel}   %
\newcommand {\dheads}     {\delta_h}         %
\newcommand {\dbodies}    {\delta_b}         %
\newcommand {\realhead}   {\sigma}           %

\newcommand {\eassign}   {\hat{\assign}}
\newcommand {\emodel}    {\hat{\pmodel}}
\newcommand {\econst}    {\hat{\pconst}}
\newcommand {\erules}    {\hat{\prules}}
\newcommand {\eorder}    {\hat{L}_<}
\newcommand {\elweight}  {\hat{\gamma}_<}
\newcommand {\eweight}   {\hat{\pweight}}

\newcommand {\ederivation} {\hat{\derivation}}
\newcommand {\edrules}    {\hat{\delta}_\prules}
\newcommand {\edatoms}     {\hat{\delta}_\pmodel}
\newcommand {\edheads}     {\hat{\delta}_h}
\newcommand {\edbodies}    {\hat{\delta}_b}
\newcommand {\erealhead}  {\hat{\realhead}}

\newcommand{\ccfont}[1]{\textsf{#1}}
\newcommand{\hy}{\hbox{-}\nobreak\hskip0pt}
\newcommand{\SB}{\{\,}
\newcommand{\SM}{\;{:}\;}
\newcommand{\SE}{\,\}}
\newcommand{\W}[1]{\ifmmode{\textnormal{\ccfont{W}\,[#1]}}\else{\textnormal{\ccfont{W}[#1]}}\fi}
\newcommand{\FPT}{\text{\normalfont FPT}}

\newcommand {\Mod} {\operatorname{Mod}}

\begin{document}
\maketitle

\begin{abstract}
Cardinality constraints or, more generally, weight constraints are well recognized as an important extension of answer-set programming. Clearly, all common algorithmic tasks related to programs with cardinality or weight constraints 
-- like checking the consistency of a program -- are intractable.
Many intractable problems in the area of knowledge representation and reasoning have been shown to become linear time tractable if the treewidth of the programs or formulas under consideration is bounded by some constant.
The goal of this paper is to apply the notion of treewidth to 
programs with cardinality or weight constraints
and to identify tractable fragments. It will turn out that the straightforward application of treewidth to such class of programs does not suffice to obtain tractability. However, by imposing further restrictions, tractability can be achieved.
\end{abstract}

\section{Introduction}

Answer-set programming (ASP) has evolved as a paradigm that allows for very elegant solutions to many combinatorial problems \cite{MarekT99}. The basic idea is to describe a problem by a logic program in such a way that the stable models correspond to the
solutions of the considered problem.
By extending logic programs with cardinality or, more generally,  weight constraints, an even larger class of problems is accessible to this method \cite{NiemelaSS99}. For instance, in the product configuration domain, we need to express cardinality, cost, and resource constraints, which are very difficult to capture using logic programs without weights.

In this paper, we restrict ourselves to \emph{normal logic programs with cardinality constraints} (PCCs, for short) or \emph{weight constraints}
(PWCs, for short). Clearly, all common algorithmic tasks related to PCCs
and PWCs -- like checking the consistency of a program -- are
intractable, since intractability even holds without such constraints.
An interesting approach to dealing with intractable problems comes from
parameterized complexity theory and is based on the following
observation: Many hard problems become tractable if some parameter that
represents a structural aspect of the problem instance is small. One
important parameter is \emph{treewidth}, which measures the ``tree-likeness''
of a graph or, more generally, of a structure. In the area of knowledge
representation and reasoning (KR\,\&\,R), many tractability results for
instances of bounded treewidth have been recently proven
\cite{GottlobPW10}.  
The goal of this work is to obtain tractability
results via bounded treewidth also for PCCs and PWCs.
Hereby, the treewidth of a PCC or PWC is defined in terms of 
its incidence graph (see Section~\ref{sect:background}).
It will turn out that the
  straightforward application of treewidth to PWCs does not suffice to
  obtain tractability. However, by imposing further restrictions,
  tractability can be achieved.

\paragraph{Main results of the paper.}
\mbox{}

\smallskip

$\bullet$ \ We show that the consistency problem of PWCs remains NP-complete even if the treewidth of the considered programs is bounded by a constant (actually, even if this constant is 1). Hence, we have to search for further restrictions on the PWCs to ensure tractability.

\smallskip

$\bullet$ \ We thus consider the largest integer occurring in (lower or
upper) bounds of the constraints in the PWC, and call this parameter
constraint-width. If also the constraint-width is bounded by an arbitrary
but fixed constant, then the consistency problem of PWCs becomes
\emph{linear time tractable} (the bound on the running time
entails a constant factor that is exponential in constraint-width and
treewidth).

\smallskip

$\bullet$ \ For PCCs (i.e., PWCs where all weights are equal to~1) we obtain  {\em non-uniform polynomial time\/} tractability by designing a new dynamic programming algorithm. Let $w$ denote the treewidth of
a PCC $\Pi$ and let $n$ denote the size of $\Pi$. Then our algorithm
works in time $\bigO(f(w) \cdot n^{2w})$ for some function $f$ that only
depends on the treewidth, but not on the size $n$ of the program. The
term ``non-uniform'' refers to the factor $n^{2w}$ in the time bound, where the size $n$ of the program is raised to the power of an expression that depends on the treewidth $w$. We shall also discuss further extensions
of this dynamic programming algorithm for PCCs. For example, it can be used to solve in non-uniform polynomial time the consistency problem of PWCs if the weights are given in unary representation.

\smallskip

$\bullet$ \ Of course, an algorithm for the PCC consistency problem that operates in time $\bigO(f(w) \cdot n^{O(1)})$
would be preferable,
i.e., the parameter $w$ does not occur in the exponent of the program size $n$. A
problem with such an algorithm is
called \emph{fixed-parameter tractable}. Alas, we show that
under common complexity theoretical assumptions no
such algorithm exists. Technically, we prove that the consistency problem of PCCs
parameterized by treewidth is hard
for the parameterized complexity class $\W{1}$. %
In other words, a non-uniform polynomial-time running time of our dynamic programming algorithm is the best that one can expect.

\paragraph{Structure of the paper.}
After recalling the necessary background in Section~\ref{sect:background}, we prove in Section~\ref{sect:np} the NP-completeness of the consistency problem of PWCs in case of binary representation of the weights.
In Section~\ref{sect:linear}, we show the linear fixed-parameter tractability of the problem if we consider the treewidth plus the size of the bounds as parameter.
In Section~\ref{sect:dp}, the non-uniform polynomial-time upper bound for the consistency problem of PCCs is established by presenting a dynamic programming algorithm.
Section~\ref{sect:extensions} contains the extensions of the dynamic programming algorithm.
By giving a $\W{1}$-hardness proof in case of unary representation in Section~\ref{sect:w1}, we show that it is unlikely that this result can be significantly improved.
Section~\ref{sect:discussion} contains a discussion and a conclusion is given in Section~\ref{sect:conclusion}.

\section{Background}
\label{sect:background}

\paragraph{Weight constraint programs.}

A \emph{program with weight constraints} (PWC)
is a triple $\prog = (\atoms,\const,\rules)$, where $\atoms$ is a set of
\emph{atoms}, $\const$ is a set of \emph{weight constraints}
(or \emph{constraints} for short), and $\rules$ is a set of \emph{rules}.
Each constraint $c\in \const$ is a triple
$(S,l,u)$ where $S$ is a set of \emph{weight literals} over $\atoms$ representing a clause  and $l\leq u$
are nonnegative integers, the lower and upper bound. A weight literal over $\atoms$ is a pair
$(a,j)$ or $(\neg a,j)$ for $a\in A$ and $1 \leq j \leq u+1$, the weight of the
literal. Unless stated otherwise, we assume that the bounds and weights are given in
binary representation.
For a constraint $c=(S,l,u) \in \const$, we write
$\clause(c) \coloneqq S$, $\low(c) \coloneqq l$, and $\upper(c) \coloneqq u$.
Moreover, we use $a \in \clause(c)$ and $\neg a \in \clause(c)$ as an abbreviation for
$(a,j) \in \clause(c)$ respectively $(\neg a,j) \in \clause(c)$ for an arbitrary $j$.
A rule $r\in \rules$ is a pair $(h,b)$ where $h\in \const$ is the head and $b\subseteq \const$ is the body.
We write $\head \coloneqq h$ and $\body \coloneqq b$.
We denote by $\size{\prog}$ the size of a reasonable encoding of program $\prog$ and call it the size of $\prog$.
Unless otherwise stated, weights are assumed to be encoded in binary notation.
For instance taking
$\size{\prog} = \card{\atoms} + \sum_{(S,l,u) \in \const} (1 + \log l + \log u + \sum_{(\mathit{lit},j) \in S} (1 + \log j)) + \sum_{(h,b) \in \rules} (1+\card{b})$
would do.
Given a constraint $c \in \const$ and an interpretation $I \subseteq \atoms$
over atoms $\atoms$,
we denote the weight of $c$ in $I$ by
\[
\weight(c,I) = \sum_{\substack{(a,j) \in \clause(c) \\ a \in I}} j~+ \sum_{\substack{(\neg a,j) \in \clause(c) \\ a \not\in I}} j\quad.
\]
$I$ is a model of $c$, denoted by $I \models c$, if
$\low(c) \leq \weight(c,I) \leq \upper(c)$.
For a set $C\subseteq \const$, $I \models C$ if $I \models c$ for all $c \in C$.
Moreover, $C$ is a model of a rule $r \in \rules$, denoted by $C \models r$, if
$\head \in C$ or $\body \not\subseteq C$.
$I$ is a model of program $\prog$ (denoted by $I \models \prog$) if
$\{ c \in \const \mid I \models c \} \models r$ for all $r \in \rules$.
If the lower bound of a constraint $c \in \const$ is missing,
we assume $\low(c) = 0$. If the upper bound
is missing, $I \models c$ if $\low(c) \leq \weight(c,I)$.
A \emph{program with cardinality constraints} (PCC) can be seen as a special case of a PWC,
where each literal has weight $1$.

\paragraph{Stable model semantics.}
Given a PWC $\prog = (\atoms,\const,\rules)$ and an
interpretation $I \subseteq \atoms$.
Following~\cite{NiemelaSS99}, the reduct $c^I$ of a constraint $c \in \const$
w.r.t.\ $I$ is obtained by removing all negative literals and the upper bound from $c$, and replacing the lower bound by
\[
l' = \max(0,\; \low(c) - \sum_{\substack{(\neg a,j) \in \clause(c) \\ a \not\in I}} j).
\]
The reduct $\prog^I$ of program $\prog$ w.r.t.\ $I$ can be obtained by
first removing each rule $r \in \rules$ which contains a
constraint $c \in B(r)$ with $W(c,I) > u(c)$.
Afterwards, each remaining rule $r$ is replaced by the set of rules\footnote{With some abuse of notation,
we sometimes write for an atom~$h$,
$(h,b)$ as a shorthand for the rule $((\{(h,1)\},1,1),b)$.} $(h,b)$, where
$h \in I \cap \clause(\head)$ and $b = \{c^I \mid c \in \body \}$,
i.e., the head of the new rules is an atom instead of a constraint.
Interpretation $I$ is called a \emph{stable model} (or \emph{answer set}) of $\prog$ if
$I$ is a model of $\prog$ and
there exists no $J \subset I$ such that $J$ is a model of $\prog^I$.
The set of all answer sets of $\prog$ is denoted by $\AS{\prog}$.
The \emph{consistency problem} for PWCs is to decide whether $\AS{\prog} \neq \emptyset$.

\paragraph{Tree decompositions and treewidth.}

A \emph{tree decomposition} of a graph $\graph=(V,E)$
is a pair $\TD=(\tree,\chi)$, where $\tree$ is a tree and
$\chi$ maps each node $n$ of $\tree$
(we use $n\in\tree$ as a shorthand below)
to a \emph{bag} $\bag(n)\subseteq V$
such that
\begin{itemize}
\item[(1)] for each $v \in V$, there is an $n \in T$ with $v \in \bag(n)$;
\item[(2)] for each $(v,w) \in E$, there is an $n \in T$ with ${v, w} \in \bag(n)$;
\item[(3)] for each
$n_1, n_2, n_3\in\tree$ such that $n_2$ lies on the path from $n_1$ to $n_3$,
    $\bag(n_1) \cap \bag(n_3) \subseteq \bag(n_2)$ holds.
\end{itemize}
\smallskip

A tree decomposition $(\tree,\chi)$ is called
\emph{normalized}
(or \emph{nice}) \cite{Kloks94},
if $\tree$ is a rooted tree and the following conditions hold:
(1)~each $n\in\tree$ has $\leq 2$ children;
(2)~for each $n\in T$ with two children $n_1,n_2$, $\bag(n)=\bag(n_1)=\bag(n_2)$;
and
(3)~for each $n\in T$ with one child $n'$,
$\bag(n)$ and $\bag(n')$ differ in exactly one element.

The \emph{width} of a tree decomposition is defined as the cardinality of its
largest bag
$\bag(n)$
minus one.
It is known that
every tree decomposition can be normalized in linear time without increasing the width \cite{Kloks94}.
The \emph{treewidth} of a graph $\graph$, denoted as $\tw(\graph)$, is the minimum
width over all tree decompositions of~$\graph$.
For arbitrary but fixed $w \geq 1$, it is feasible in linear time
to decide whether a graph
has treewidth $\leq w$ and, if so, to
compute a tree decomposition of width $w$,
see \cite{Bod96}.

\paragraph{Treewidth and constraint-width of PWCs.}

To build tree decompositions for programs, we
use \emph{incidence graphs}.
For a PWC $\prog = (\atoms,\const,\rules)$,
such a graph has vertex set
$\atoms \cup \const \cup \rules$.
There is an edge between $a \in \atoms$ and $c \in \const$ if $a \in \clause(c)$ or
$\neg a \in \clause(c)$, and there is an edge between
$c \in \const$ and $r  \in \rules$ if $c \in \{\head\} \cup \body$.
The treewidth of $\prog$, denoted by $\tw(\prog)$, is the treewidth of its incidence graph.
The \emph{constraint-width} of $\prog$,
denoted by
$\ww(\prog)$, is the largest (lower or upper) bound occurring in the
constraints of $\const$ (or 0 if there are no bounds).

\begin{example}
\label{ex:config}
Consider the following system configuration problem,
where one has to choose among the given parts:
$p_1: 4000\$$, $p_2: 2000\$$, and $p_3: 1000\$$ such that
the total cost is $\leq 5000\$$. Thereby one of
$\{p_1,p_2\}$ has to be selected and $p_3$ requires $p_2$.

This scenario can be represented by the PWC
\[\exprog=(\{p_1,p_2,p_3\},\{c_1,c_2,c_3,c_4\},\{r_1,r_2,r_3\})\]
with
\begin{align*}
c_1& =(\{(p_1,4),(p_2,2),(p_3,1)\},0,5)& r_1& =(c_1,\emptyset)\\
c_2& =(\{(p_1,1),(p_2,1)\},1,2)& r_2& =(c_2,\emptyset)\\
c_3& =(\{(p_2,1)\},1,1)& r_3& =(c_3,\{c_4\})\\
c_4& =(\{(p_3,1)\},1,1)
\end{align*}

\smallskip

\noindent
The incidence graph $G_{Ex}$ of $\exprog$ as well as
a normalized tree decomposition $\TD_{Ex}$ for $\exprog$ of width 2
are depicted in Figure~\ref{fig:td}.
\end{example}

\begin{figure}[t]
\centering
\begin{tikzpicture}[node distance=1.5cm]
  \tikzstyle{every node}=[draw,circle,fill,inner sep=2pt]
  \node (c1) [label=left:$c_1$] {};
  \node (c2) [below of=c1,label=left:$c_2$] {};
  \node (c3) [below of=c2,label=left:$c_3$] {};
  \node (c4) [below of=c3,label=left:$c_4$] {};
  \node (p1) [label=right:$p_1$] at ++(0.75,-0.75) {};
  \node (p2) [below of=p1,label=right:$p_2$] {};
  \node (p3) [below of=p2,label=right:$p_3$] {};
  \node (r1) [label=left:$r_1$] at ++(-0.75,-0.75) {};
  \node (r2) [below of=r1,label=left:$r_2$] {};
  \node (r3) [below of=r2,label=left:$r_3$] {};
  \node [draw=none,fill=none,] at ++(-1,0.5) {$G_{Ex}$:};
  \draw [-] (p1) -- (c1);
  \draw [-] (p1) -- (c2);
  \draw [-] (p2) -- (c1);
  \draw [-] (p2) -- (c2);
  \draw [-] (p2) -- (c3);
  \draw [-] (p3) -- (c1);
  \draw [-] (p3) -- (c4);
  \draw [-] (r1) -- (c1);
  \draw [-] (r2) -- (c2);
  \draw [-] (r3) -- (c3);
  \draw [-] (r3) -- (c4);
\end{tikzpicture}
\hspace{0.5cm}
\begin{tikzpicture}
\tikzstyle{every node}=[draw,rectangle, label distance=5mm]
\tikzstyle{level 1}=[level distance=0.75cm, sibling distance=3.2cm]
\tikzstyle{level 5}=[level distance=0.75cm, sibling distance=2.6cm]
\tikzstyle{level 6}=[level distance=0.75cm, sibling distance=1.8cm]
\node (root) [label=left:$\TD_{Ex}$:] {$p_2,c_1,c_3$}
  child {node {$p_2,c_1,c_3$}
    child {node {$c_1,c_3$}
      child {node {$p_3,c_1,c_3$}
        child {node {$p_3,c_3$}
          child {node {$p_3,c_3,c_4$}
            child {node {$c_3,c_4$}
              child {node {$c_3,c_4,r_3$}}}}}}}}
  child {node {$p_2,c_1,c_3$}
    child {node {$p_2,c_1$}
      child {node {$p_2,c_1,c_2$}
        child {node {$c_1,c_2$}
          child {node {$c_1,c_2$}
            child {node {$p_1,c_1,c_2$}}}
          child {node {$c_1,c_2$}
            child {node {$c_1,c_2$}
              child {node {$c_1,c_2,r_1$}}}
            child {node {$c_1,c_2$}
              child {node {$c_1,c_2,r_2$}}}}}}}};
\end{tikzpicture}
\caption{Incidence graph $G_{Ex}$ and tree decomposition $\TD_{Ex}$ of Example~\ref{ex:config}.}
\label{fig:td}
\end{figure}
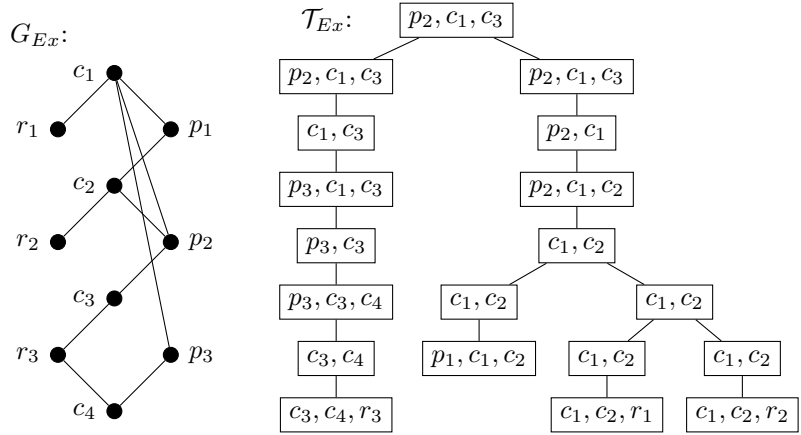

\section{NP-Completeness}
\label{sect:np}

\begin{theorem}\label{the:nphard}
  The consistency problem for PWCs is NP-complete already for programs having
  treewidth~1.
\end{theorem}

\begin{proof}
  Clearly the problem is in NP\@.
  To show NP-hardness we reduce from the well-known NP-complete problem \textsc{Partition}.
  An instance of \textsc{Partition} is a collection of positive integers $X=\{x_1,\dots,x_n\}$ (encoded in binary); the question is whether there exists a set $I\subseteq \{1,\dots,n\}$ such that $\sum_{i\in I} x_i = \sum_{i\notin I} x_i$.
  Given an instance $X=\{x_1,\dots,x_n\}$, we construct a PWC $\prog=(\atoms,\const,\rules)$ as follows.
  Let $S=\sum_{i=1}^n x_i$; we may assume that $S$ is even since otherwise $X$ is a no-instance and can immediately be rejected.
  We put $A=\{a_1,\dots,a_n\}$, $\const=\{c\}$ where $c=(\{ (a_1,x_1),\dots, (a_n,x_n) \}, S/2,S/2)$, and $\rules=\{(c,\emptyset)\}$.

  \claim{1}{$\prog$ has treewidth 1.}
  By construction the incidence graph of $\prog$ is a tree, hence of treewidth~1.

  \claim{2}{$X$ is a yes-instance of \textsc{Partition} if and only if $\prog$ has a model.}
  This claim follows easily from the definitions.

  \claim{3}{All models of $\prog$ are stable.}
  Let $M$ be a model of~$\prog$.
  Since each atom appears positively in a constraint at the head of a rule, and since all the rules have an empty body, it follows that the reduct $\prog^{M}$ is the conjunction of all the elements of $M$.
  Hence $M$ is stable since no proper subset of $M$ can satisfy $\prog^{M}$.
  We conclude that $X$ is a yes-instance of \textsc{Partition} if and only if $\prog$ is consistent.

  It is evident that $\prog$ can be constructed from $X$ in polynomial time.
  Hence, by Claims 1--3 we have a polynomial-time reduction from \textsc{Partition} to the consistency problem of PWCs of treewidth~1, and the theorem follows.
\end{proof}

Note that \textsc{Partition} is ``weakly NP-hard'' since its NP-hardness
depends on the binary encoding of the
given
integers. Accordingly, our
reduction provides only weak NP-hardness for the consistency of
PWCs of bounded treewidth.
In fact, we shall prove in Section~\ref{sect:extensions}
that if we assume the weights to be given in unary the consistency
problem is feasible in (non-uniform) polynomial time for PWCs of bounded
treewidth.

\section{Linear-Time Tractability}
\label{sect:linear}

\begin{theorem}\label{the:lintime}
  The consistency problem for PWCs can be solved in linear time for instances
  whose treewidth and constraint-width are bounded by constants.
\end{theorem}

To prove this result we shall take a logic approach and use Courcelle's
Theorem \cite{Courcelle87}, see also
\cite{DowneyFellows99,FlumGrohe06}. To this aim we consider
\emph{Monadic Second Order} (MSO) logic on labeled graphs in terms of
their incidence structure whose universe contains vertices and edges. We
assume an infinite supply of individual variables $x,x_1,x_2,\dots$ and
set variables $X,X_1,X_2,\dots$\ The atomic formulas are $E(x)$ (``$x$
is an edge''), $V(x)$ (``$x$ is a vertex''), $I(x,y)$ (``vertex $x$ is
incident with edge~$y$''), $x=y$ (equality), and $X(y)$ (``element $y$
belongs to set~$X$'').  Further we assume that a vertex or edge $x$ can
be labeled with an element $a$ of some fixed finite set, denoted by the
atomic formula $P_a(x)$.  MSO formulas are built up from atomic
formulas using the usual Boolean connectives
$(\neg,\wedge,\vee)$, quantification over
individual variables ($\forall x$, $\exists x$), and quantification over
set variables ($\forall X$, $\exists X$).

We write $G \models \phi$ to indicate that an MSO formula $\phi$ is true
for the labeled graph~$G$.
Courcelle's Theorem states that $G \models \phi$ can be checked in linear time for labeled graphs if a tree decomposition of constant width is provided as an input.
The latter is no restriction for proving Theorem~\ref{the:lintime}, since by Bodlaender's Theorem~\cite{Bod96}, we can compute in linear time a tree decomposition of smallest width for graphs whose treewidth is bounded by a constant.

\newcommand{\lbound}{\text{low}}
\newcommand{\ubound}{\text{up}}

Let $k$ be a constant and consider a PWC $\prog = (\atoms,\const,\rules)$ of
constraint-width $k$.  We encode all the information of $\prog$ by adding edge and
vertex labels to the incidence graph of $\prog$.  We use the edge labels $+,-$
to indicate polarity of literals and the labels $h,b$ to distinguish between
head and body of rules.  That is, an edge $\{a,c\}$ for $a\in A$ and $c\in \const$
has label $+$ if $a\in \clause(c)$, and label $-$ if $\neg a\in \clause(c)$;
an edge $\{c,r\}$ for $c\in \const$ and $r\in \rules$ has label $h$ if $c=\head$
and label $b$ if $c\in \body$.  We use edge labels $1,\dots,k+1$ to encode
weights of literals (literals of weight 0 can be omitted, weights exceeding
$k+1$ can be replaced by $k+1$). That is, an edge $\{a,c\}$ for $a\in A$ and $c\in
\const$ has label $j$ if the constraint $c$ contains the weight literal
$(a,j)$ or $(\neg a,j)$.  We use vertex labels $\lbound[i]$ for $i\in
\{0\dots,k\}$ and $\ubound[j]$ for $j\in \{0\dots,k,\infty\}$ to encode the
bounds of constraints (we use $\lbound[0]$ and $\ubound[\infty]$ in case the
lower or upper bound is missing, respectively).
Finally we use vertex labels
$\atoms,\const,\rules$ to indicate whether a vertex represents an atom, a
clause or a rule, respectively.

Let $G$ denote the incidence graph of the PWC $\prog$ with added labels as
described above.  In the following we will explain how to construct an MSO
formula $\phi$ such that $G \models \phi$ if and only if $\prog$ has a stable model. For
convenience we will slightly abuse notation and use meta-language terms as
shorthands for their obvious definitions in the MSO language; for example we
will write $X\subseteq Y$ instead of $\forall x (X(x)\rightarrow Y(x))$, and
$a\in \atoms$ instead of $V(a) \wedge P_\atoms(a)$.

Let $X$ and $Y$ be set variables and $c$ an individual variable.
For each integer $s \in \{0,\dots, k+1\}$ we define an MSO formula
$\text{Sum}_s(X,Y,c)$ that is true for $G$ if and only if $X$ and $Y$ are
interpreted as sets of atoms, $c$ is interpreted as a constraint, and we
have
\[s=
\sum_{\substack{(a,j)\in \clause(c)\\ a\in X}} j+ \sum_{\substack{(\neg
  a,j)\in \clause(c)\\ a\notin Y}} j.\]
We use the fact that it is always sufficient to choose at most $k+1$ literals
from $c$ (say $r$ positive and $r'$ negative literals) to witness that the
above equality holds.

\medskip\noindent $\text{Sum}_s(X,Y,c)\equiv$

$ X,Y\subseteq \atoms \wedge c\in
\const$ \hfill {\footnotesize (1)}

$\wedge \bigvee_{1\leq r +r' \leq  k,\
         1\leq n_1,\dots,n_{r+r'} \leq k+1,\
         s= n_1+\dots + n_{r+r'} }
\exists e_1,\dots, e_{r+r'}$ \hfill {\footnotesize (2)}

~~$\big[\bigwedge_{i=1}^{r+r'} (P_{n_i}(e_i) \wedge I(c,e_i) \wedge \exists a\in
A, I(a,e_i))$ \hfill {\footnotesize (3)}

~~~$\wedge \bigwedge_{1\leq i< i' \leq r+r'}
e_i\neq e_{i'}$ \hfill {\footnotesize (4)}

~~~$\wedge \mathop{\forall e\in E}\; ( \neg I(c,e) \vee  \forall a\in
\atoms, \neg I(a,e) \vee \bigvee_{i=1}^{r+r'} e=e_i)$
\hfill {\footnotesize (5)}

~~~$\wedge
\bigwedge_{i=1}^r ( P_{+}(e_i) \wedge  \exists a\in X, I(a,e_i))$
\hfill {\footnotesize (6)}

~~~$\wedge
\bigwedge_{i=r+1}^{r'} ( P_{-}(e_i) \wedge \neg \exists a\in Y,
I(a,e_i)) \; \big]$
\hfill {\footnotesize (7)}

\smallskip\noindent Some further explanation: Each of the $r+r'$ literals is
represented by an edge $e_i$ of weight $n_i$.  The disjunction in line~(2)
runs over all possible combinations of weights $n_1,\dots,n_{r+r'}$ that give
the sum~$s$.  Line~(3) makes sure that each edge $e_i$ has weight $n_i$ and
runs between constraint $c$ and some atom.  Lines~(4) and (5) make sure that
the edges are mutually different and that no other edge runs between
constraint $c$ and an atom. Lines~(6) and (7) make sure that $e_1,\dots,e_r$
represent positive literals over atoms that belong to~$X$, and
$e_{r+1},\dots,e_{r+r'}$ represent negative literals over atoms that do not
belong to~$Y$.

\smallskip \noindent The following formula is true if and only if $X$ satisfies
$c$.
\begin{align*}
\text{Sat}(X,c) &\equiv \text{SatL}(X,X,c)\wedge \text{SatU}(X,X,c), \text{ where} \\
\text{SatL}(X,Y,c) &\equiv P_{\lbound[0]} \vee \bigvee_{i\in \{1,\dots,k\}}  (P_{\lbound[i]}(c)
\wedge \bigvee_{i \leq s \leq k+1} \text{Sum}_s(X,Y,c)), \text{ and} \\
\text{SatU}(X,Y,c) &\equiv P_{\ubound[\infty]} \vee \bigvee_{j \in \{0,\dots,k\}}   (P_{\ubound[j]}(c)
\wedge \bigvee_{0 \leq s \leq j} \text{Sum}_s(X,Y,c)).
\end{align*}

\smallskip \noindent The next formula is true if and only if $Y$ is a model
of~$\prog$.
\begin{align*}
  \Mod(Y) &\equiv \forall r\in \rules\; \exists c\in \const\; \big[ (H(c,r) \wedge
\text{Sat}(Y,c)) \vee (B(c,r) \wedge \neg \text{Sat}(Y,c))\big], \text{ where} \\
H(c,r) &\equiv \exists e\in E\; (I(c,e) \wedge I(r,e) \wedge P_h(e)), \text{ and} \\
B(c,r) &\equiv \exists e\in E\; (I(c,e) \wedge I(r,e) \wedge P_b(e)).
\end{align*}

\smallskip \noindent Finally, the formula $\text{SMod}(Y)$ is true if and only if
$Y$ is a stable model of $\prog$. We make use of the formula
$\text{Red}(X,Y)$ that states that $X$ satisfies the
reduct~$\prog^Y$.
\begin{align*}
\text{SMod}(Y) &\equiv \Mod(Y)
\wedge \forall X\subseteq Y\; (X=Y \vee \neg \text{Red}(X,Y)), \text{ where} \\
\text{Red}(X,Y) &\equiv \forall r\in\rules\;
\forall a\in \atoms\; [
a\in X \vee a\notin Y \vee \neg \text{InH}(a,r) \\
&\quad \vee \exists c\; (B(c,r) \wedge (
    \neg \text{SatU}(Y,Y,c) \vee
    \neg \text{SatL}(X,Y,c) ))], \text{ and} \\
\text{InH}(a,r) &\equiv \exists c\in \const\; \exists e, e'\in E\;
    [I(a,e) \wedge I(c,e) \wedge P_+(e)\\
    &\quad \wedge I(r,e') \wedge I(c,e') \wedge P_h(e')],
\end{align*}
that is, $a$ is an atom that occurs as a positive literal in the
constraint at the head of rule~$r$.

We summarize the correctness of the construction in the following lemma.

\begin{lemma}
  Let $\phi=\exists Y\; \text{\normalfont SMod}(Y)$. Then $\prog$ has a
  stable model if and only if $G \models \phi$.
\end{lemma}
Since the labeled graph $G$ can be constructed in linear time, Theorem~\ref{the:lintime} now follows directly by Courcelle's Theorem.

\section{Dynamic Programming Approach}
\label{sect:dp}

Recently,
\cite{JaklPW09}
presented a dynamic programming algorithm for answer-set programming
that works for programs without cardinality or weight constraints, but
possibly with disjunction in the head of the rules.
One way to obtain a dynamic programming algorithm for PCCs is to try to
extend that algorithm of Jakl et al.\
by methods to handle the cardinality
constraints. In principle, this should be feasible. However, computationally, this approach has a serious drawback, namely: the aforementioned algorithm is
tractable for bounded treewidth, but it is {\em double exponential\/}
w.r.t.\ the treewidth (basically this is due to the handling of disjunctions).
Our goal here is to present an algorithm that is only {\em single exponential\/} w.r.t.\ the treewidth. In order to achieve this goal, we have to manipulate a slightly more complicated data structure along the bottom-up traversal of the tree decomposition. In particular, we have to deal with orderings on the atoms in a model.

To this end, we need an alternative characterization of
stable models.
Slightly rephrasing a result by
\cite{Liu09}
we can characterize answer sets of PCCs as follows:

\begin{proposition}
\label{prop:as}
Given a PCC $\prog=(\atoms,\const,\rules)$,
$\pmodel \subseteq \atoms$ is an answer set (stable model) of $\prog$
if and only if
the following conditions are jointly satisfied:
\begin{itemize}
\item $\pmodel$ is a model of $\prog$, i.e., $\pmodel \models \prog$,
\item there exists a strict linear order $<$ over $\pmodel$, such that
for each atom $a \in \pmodel$, there exists a rule $r\in \rules$ with\\
(R1)\ $a \in \clause(\head)$,\\
(R2)\ $\pmodel \models \body$,\\
(R3)\ for each $c\in \body$, $\low(c)\leq
\card{\{b \in \clause(c) {\mid} b < a\}\cup\{ \neg b \in \clause(c) {\mid} b \in \atoms \setminus \pmodel\}}$.
\end{itemize}
\end{proposition}

Since the handling of linear orders is crucial
for utilizing the above characterization, we will fix some notation first.
We denote by $[x_1,x_2,\ldots,x_n]$ a (strict) linear order
$x_1 < x_2 < \ldots < x_n$ on a set $X=\{x_1,\ldots,x_n\}$.
Moreover, $\order{X}$ denotes the set of all possible linear orders over $X$.
Two linear orders $[x_1,\ldots,x_n]$ and $[y_1,\ldots,y_m]$
are called \emph{inconsistent}, if there are $x_i,x_j,y_k,y_l$ such that
$x_i<x_j$, $y_k<y_l$, $x_i=y_l$ and $x_j=y_k$.
Otherwise, we call them \emph{consistent}.
Given two consistent linear orders $[x_1,\ldots,x_n] \in \order{X}$ and $[y_1,\ldots,y_m] \in \order{Y}$,
we denote by $[x_1,\ldots,x_n]+[y_1,\ldots,y_m]=S$
the set of their possible \emph{combinations}. $S$ contains those linear orders $[z_1,\ldots,z_p] \in \order{X\cup Y}$
such that for every pair $x_i < x_j$ (respectively $y_i < y_j$), there exists $z_k < z_l$ with
$z_k = x_i$ and $z_l = x_j$ (respectively $z_k = y_i$ and $z_l = y_j$).
Note that in general, there exists more than one possible combination.
Furthermore, we denote by $[x_1,\ldots,x_{i-1},x_i,x_{i+1},\ldots,x_n] - [x_i]$
the linear order $[x_1,\ldots,x_{i-1},x_{i+1},\ldots,x_n]$.

Throughout the whole section, let $\TD=(\tree,\bags)$ be a
normalized tree decomposition of a PCC $\prog = (\atoms,\const,\rules)$.
We present a dynamic programming algorithm,
traversing $\TD$ in bottom-up direction 
in order to compute whether $\prog$
admits an answer set. Ultimately, we will state properties
about subtrees of $\TD$ and inductively add more and more nodes,
until we get a statement about the whole tree. To this end, the following
notions become handy.
Given a node $n \in \tree$,
we denote by $\tree_n$ the subtree of $\tree$ rooted at $n$.
For a set $S \subseteq \atoms \cup \const \cup \rules$,
$\innode{n}{S}$ is a shorthand for $\bags(n) \cap S$.
Moreover, $\intree{n}{S} \coloneqq \bigcup_{m\in \tree_n} \innode{m}{S}$
and $\insubtree{n}{S} \coloneqq \intree{n}{S} \setminus \innode{n}{S}$.
Since the scope of a solution will always be limited to a subtree of the whole tree decomposition, the notion of
a model has to be refined with respect to a universe $U=\intree{n}{A}$.
To this end, the cardinality
of a constraint $c \in \const$ with respect to an interpretation
$I \subseteq \univ$ is given by
\[
\car(c,I,\univ) = \card{\{b \in \clause(c) \mid b \in I\}}+
\card{\{\neg b \in \clause(c) \mid b \in \univ \setminus I\}}.
\]
Then $I$ is a model of $c$ under universe $\univ$ (denoted by
$I \models_{\univ} c$) if $\low(c) \leq \car(c,I,\univ) \leq \upper(c)$.
Note that $\models_{\univ}$ and $\models$ coincide for $\univ = \atoms$.
Similarly, for a subset of constraints $\const' \subseteq \const$,
set $\pconst \subseteq \const'$ is a model of a rule $r \in \rules$ under restriction $\const'$,
denoted by $\pconst \models_{\const'} r$,
if $\head \in \pconst$ or $\body \cap \const' \not\subseteq \pconst$.

In order to facilitate the
discussion below, we define the following sum for constraint
$c \in \const$, interpretation $I \subseteq \univ$ over a set of
atoms $\univ \subseteq \atoms$ and linear order $\porder$ containing at least $I \cup \{c\}$:
\begin{eqnarray*}
\car_<(c,I,\univ,\porder) &=& \card{\{b \in \clause(c) \mid b \in I \wedge b < c\}} +\\
&& \card{\{\neg b \in \clause(c) \mid b \in \univ\setminus I\}}.
\end{eqnarray*}

The following definition helps us to find partial answer sets,
limited to the scope of a subtree of $\TD$.

\begin{definition}
\label{def:p-solution}
A \emph{partial solution} (for node $n \in \tree$) is a tuple
$\eassign=(n,\emodel,\econst,\erules,\eorder,\eweight,\elweight,\ederivation)$, with
interpretation $\emodel \subseteq \intree{n}{\atoms}$,
satisfied constraints $\econst \subseteq \intree{n}{\const}$,
satisfied rules $\erules \subseteq \intree{n}{\rules}$,
linear order $\eorder \in \order{\emodel \cup \econst \cup \intree{n}{\rules}}$,
cardinality functions $\eweight : \intree{n}{\const} \ra \N$ and
$\elweight : \econst \ra \N$, and
derivation witness $\ederivation = (\edrules,\edatoms,\edheads,\edbodies,\erealhead)$ with
derivation rules $\edrules \subseteq \intree{n}{\rules}$,
derived atoms $\edatoms \subseteq \emodel$,
derivation head constraints $\edheads \subseteq \econst$,
derivation body constraints $\edbodies \subseteq \econst$, and
check function $\erealhead : \edheads \ra \{0,1\}$ such that
the following conditions are jointly satisfied:
\begin{enumerate}
\item $\econst \cap \insubtree{n}{\const} = \{c \in \insubtree{n}{\const} \mid \emodel \models_{\sintree{n}{\atoms}} c\}$
\item $\erules = \{r \in \intree{n}{\rules} \mid \econst \models_{\sintree{n}{\const}} r\}$ and
$\insubtree{n}{\rules} \subseteq \erules$
\item $\eweight(c) = \car(c,\emodel,\intree{n}{\atoms})$ for all $c \in \intree{n}{\const}$
\item $\elweight(c) = \car_<(c,\emodel,\intree{n}{\atoms},\eorder)$ for all $c \in \econst$
\item $\edatoms = \{a \in \emodel \mid c \in \edheads, a \in \clause(c), a > c\}$ and
$\emodel \cap \insubtree{n}{\atoms} \subseteq \edatoms$
\item $\edbodies = \bigcup_{r \in \edrules} \body$ and $\edbodies \subseteq \econst$
\item $c \in \body \Ra r > c$ for all $c \in \edbodies$ and $r \in \edrules$
\item $\low(c) \leq \elweight(c)$ for all $c \in \edbodies \cap \insubtree{n}{\const}$
\item $\erealhead(c) = 1 \LR \exists r \in \edrules$ with $\head = c$ and $c > r$
\item $\erealhead(c) = 1$ for all $c \in \edheads \cap \insubtree{n}{\const}$
\end{enumerate}
\end{definition}

The idea of this data structure is that, for
some atom, clause, or rule that is no longer ``visible'' in the current
bag but was included in the subtree, the containment in one of the sets of $\eassign$
is strictly what one would expect from an answer set, while
for elements that are still visible this containment does not have
to fulfill that many conditions and can be seen as some sort of ``guess''.
For example, $\econst \cap \insubtree{n}{\const}$, the set of constraints in $\econst$
that are no longer visible, indeed contains exactly the constraints that are
satisfied under interpretation $\emodel$, i.e.,
$\{c \in \insubtree{n}{\const} \mid \emodel \models_{\sintree{n}{\atoms}} c\}$,
while $\econst \cap \innode{n}{\const}$
represents the guess of those constraints, we still want to become true when we
further traverse the tree towards the root node.
$\emodel, \econst, \erules,$ and $\eweight$ are used to ensure that the answer
set is a model of our program. $\eorder$ is the strict linear order, whose existence
is demanded in the definition of answer sets.
$\elweight$ will be used to check condition (R3) of stable models, i.e.,
it will contain the cardinality on the left side of the equation in (R3).
The derivation of atoms $a \in \emodel$ is represented by $\ederivation$.
The definition of answer sets requires for each $a \in \emodel$ the existence of
some rule $r \in \rules$ satisfying (R1)-(R3). The set of those rules
will be represented by $\edrules$.
Sets $\edheads$ and $\edbodies$ contain the head, and respectively, body constraints
of the rules in $\edrules$.
The set $\edatoms$ contains those atoms, for which we already found
a head constraint to derive it. $\erealhead$ is a utility function, which ensures that each
(guessed) constraint in $\edheads$ is indeed the head of some rule in $\edrules$.
Thereby $\erealhead(c) = 1$ marks that such a rule was found.

Note that, w.l.o.g., we may assume that
the root node of a normalized tree decomposition has an empty
bag. Indeed, this
can always be achieved by introducing at most $\tw(\prog)+1$
additional nodes above the root of a given tree decomposition.
Then the following proposition shows the correspondence between
answer sets
and
partial
solutions
for the root node of a given normalized tree decomposition.

\begin{proposition}
\label{prop:consistency}
Let $\rootnode$ be the root node of $\tree$ and let $\bag(\rootnode) = \emptyset$.
Then $\AS{\prog} \neq \emptyset$ if and only if there exists a partial solution $\eassign=(\rootnode,\emodel,\econst,\erules,\eorder,\eweight,\elweight,\ederivation)$ for $\rootnode$.
\end{proposition}

\begin{proof}
($\Rightarrow$) Given an answer set $\pmodel \in \AS{\prog}$, we construct a partial solution $\eassign$ for $\rootnode$ with derivation witness $\ederivation = (\edrules,\edatoms,\edheads,\edbodies,\erealhead)$ as follows.
Let $\emodel \coloneqq \pmodel$, let $\econst \coloneqq \{c \in \const : \pmodel \models c\}$ and let $\erules \coloneqq \rules$.
Let $\porder \coloneqq [a_1,\dots,a_{\card{\pmodel}}] \in \order{\pmodel}$ be the linear order from Proposition~\ref{prop:as} and let $f: M \ra \rules$ be the function that assigns each atom $a \in \pmodel$ the rule $r \in \rules$ that satisfies conditions (R1)--(R3) of Proposition~\ref{prop:as} for $a$.
Furthermore, let $\edrules \coloneqq \{f(a) : a \in \pmodel \}$.
In order to create $\eorder$, we modify $\porder$ as follows.
For every $r \in \edrules$ let $a_r$ be the smallest atom in $\porder$ such that $f(a_r) = r$.
Atom $a_r$ is then replaced in $\porder$ by the sequence $c_1,\dots,c_j,r,c_{j+1},a_r$, where $\{c_1,\dots,c_j\} = \body$ and $c_{j+1} = \head$.
Note that by construction $\{c_1,\dots,c_{j+1}\} \subseteq \econst$.
The remaining clauses from $\econst$ as well as the rules $\rules \setminus \prules$ are arbitrarily appended at the end of $\eorder$.
For every constraint $c \in \const$ we set $\eweight(c) \coloneqq \car(c,\pmodel,\atoms)$.
For every constraint $c \in \econst$ we set $\elweight(c) \coloneqq \car_<(c,\pmodel,A,\eorder)$.
Let $\edatoms \coloneqq \pmodel$,
let $\edheads \coloneqq \{\head : r \in \edrules\}$, and
let $\edbodies \coloneqq \bigcup_{r \in \edrules} \body$.
Finally, let $\erealhead(c) \coloneqq 1$ for all $c \in \edheads$.
We show now that $\eassign$ is indeed a partial solution by checking conditions 1--10 of Definition~\ref{def:p-solution}.
Conditions 1--4, 6--7, and 9--10 are satisfied by construction.
For each $a \in \pmodel$ let $c_a \coloneqq H(f(a))$.
Then $c_a \in \edheads$, $a \in \clause(c_a)$, and $c_a < a$.
Therefore, $\edatoms = \{a \in \emodel \mid c \in \edheads, a \in \clause(c), a > c\}$ which satisfies condition 5.
Condition 8 is satisfied because of (R3) of Proposition~\ref{prop:as}.
Hence $\eassign$ is a partial solution for $\rootnode$.

($\Leftarrow$)
For the other direction, the requirement that $\bag(\rootnode) = \emptyset$ ensures, that the guessing part of a given partial solution $\eassign$ is nonexistent.
Therefore, $\econst = \{c \in \const : \emodel \models c\}$ and $\erules = \{r \in \rules \mid \econst \models r\}= \rules$.
This ensures that $\emodel \models \prog$ and is therefore a model of $\prog$.
Let the linear order $\porder$ be the restriction of $\eorder$ to the set $\emodel$.
Let $a \in \emodel$ be an arbitrary atom.
By condition 5 of Proposition~\ref{prop:as} there exists a constraint $c \in \edheads$ with $a \in \clause(c)$ and $a > c$.
Therefore, by condition 9 and 10 there exists a rule $r \in \edrules$ with $\head = c$ and $c > r$.
We now show that rule $r$ is the one fulfilling (R1)--(R3) of Proposition~\ref{prop:as} for atom $a$.
(R1) is satisfied by construction.
By condition 6, $\body \subseteq \econst$.
Therefore, $\emodel \models \body$, satisfying (R2).
Finally, (R3) is satisfied through condition 8.
This shows that $\emodel$ is indeed an answer set of $\prog$.
\end{proof}

An algorithm that computes all partial solutions at each node of the tree decomposition is highly inefficient, since the size and the number of such solutions can grow exponentially in the input size.
Therefore we introduce \emph{bag assignments}, which is a data structure similar to partial solutions, but instead of ranging over the whole subtree, their scope is restricted to a single bag of the tree decomposition.
But we are not interested in arbitrary bag assignments.
Instead we consider only those, which can be seen as the projection of a partial solution for node $n$ to the bag of node $n$. Formally this is stated as follows:

\begin{definition}
\label{def:assign}
A \emph{bag assignment (for node $n \in \tree$)} is a tuple $\assign=(n,\pmodel,\pconst,\prules,\porder,\pweight,\lweight,\derivation)$, with partial model $\pmodel \subseteq \innode{n}{\atoms}$, satisfied constraints $\pconst \subseteq \innode{n}{\const}$, satisfied rules $\prules \subseteq \innode{n}{\rules}$, linear order $\porder \in \order{\pmodel \cup \pconst \cup \innode{n}{\rules}}$, cardinality functions $\pweight : \innode{n}{\const} \ra \N$ and $\lweight : \pconst \ra \N$, and derivation witness $\derivation = (\drules,\datoms,\dheads,\dbodies,\realhead)$ with derivation rules $\drules \subseteq \innode{n}{\rules}$, derived atoms $\datoms \subseteq \pmodel$, derivation head constraints $\dheads \subseteq \pconst$, derivation body constraints $\dbodies \subseteq \pconst$, and check function $\realhead : \dheads \ra \{0,1\}$.
\end{definition}

\begin{definition}
\label{def:model}
A bag assignment $\assign$ for node $n$ with $\assign=(n,\pmodel,\pconst,\prules,\porder,\pweight,\lweight,\derivation)$ and $\derivation = (\drules,\datoms,\dheads,\dbodies,\realhead)$ is called a \emph{bag model (for node $n$)} if there exists a partial solution $\eassign=(n,\emodel,\econst,\erules,\eorder,\eweight,\elweight,\ederivation)$, with $\ederivation = (\edrules,\edatoms,\edheads,\edbodies,\erealhead)$ such that
\begin{itemize}
\item $\emodel \cap \bag(n) = \pmodel$,\quad
      $\econst \cap \bag(n) = \pconst$,\quad
      $\erules \cap \bag(n) = \prules$,
\item $\eorder$ and $\porder$ are consistent,
\item $\eweight(c) = \pweight(c)$,\quad
      $\elweight(c) = \lweight(c)$\quad
      for all $c \in \innode{n}{\const}$,
\item $\edrules \cap \bag(n) = \drules$,\quad
      $\edatoms \cap \bag(n) = \datoms$,
\item $\edheads \cap \bag(n) = \dheads$,\quad
      $\edbodies \cap \bag(n) = \dbodies$,
\item $\erealhead(c) = \realhead(c)$\quad for all $c \in \dheads$.
\end{itemize}
\end{definition}

\noindent
Indeed, it turns out that it is sufficient to maintain only bag models
during the tree traversal.

\begin{proposition}
\label{prop:rootmodel}
Let $\rootnode$ be the root node of $\tree$, and
let $\bag(\rootnode) = \emptyset$.
Then $\AS{\prog} \neq \emptyset$ if and only if
$\assign=(\rootnode,\emptyset,\emptyset,\emptyset,[],\emptyset,\emptyset,\derivation)$
with $\derivation = (\emptyset,\emptyset,\emptyset,\emptyset,\emptyset)$
is a bag model for $\rootnode$.
\end{proposition}

\begin{proof}
Since $\bag(\rootnode) = \emptyset$, every partial solution
for $\rootnode$ is an extension of $\assign$ according to the conditions of Definition~\ref{def:model}.
Therefore, this statement follows from Proposition~\ref{prop:consistency}.
\end{proof}

By the same argument as for the root node, we may assume that $\bag(n) = \emptyset$ for leaf nodes $n$.
Now a dynamic programming algorithm can be achieved, by creating the only possible bag model $\assign=(n,\emptyset,\emptyset,\emptyset,[],\emptyset,\emptyset,\derivation)$ with $\derivation = (\emptyset,\emptyset,\emptyset,\emptyset,\emptyset)$ for each leaf $n$, and then propagating these bag models
along the paths to the root node.
Thereby the bag models are altered according to rules, which depend only on the bag of the current node.
In order to sketch the cornerstones of the dynamic programming algorithm more clearly, 
we distinguish between eight types of nodes in the tree decomposition:
leaf (L), branch (B), atom introduction (AI), atom removal (AR), rule introduction (RI), rule removal (RR), constraint introduction (CI), and constraint removal (CR) node.
The last six types will be often augmented with the element $e$ (either an atom, a rule, or a constraint) which is removed or added compared to the bag of the child node.

Next we define a relation $\prec_\TD$ between bag assignments,
which will be used to propagate bag models in a bottom-up direction
along the tree decomposition $\TD$.
Afterwards we demonstrate the intuition of these rules with the help of a small example.

\begin{definition}
\label{def:rules}
Let $\assign=(n,\pmodel,\pconst,\prules,\porder,\pweight,\lweight,\derivation)$ and $\assign'=(n',\pmodel',\pconst',\prules',\porder',\pweight',\lweight',\derivation')$ with $\derivation = (\drules,\datoms,\dheads,\dbodies,\realhead)$ and $\derivation' = (\drules',\datoms',\dheads',\dbodies',\realhead')$ be bag assignments for nodes $n,n' \in \tree$.
We relate $\assign' \prec_\TD \assign$ if $n$ has a single child $n'$ and the following properties are satisfied, depending on the node type of $n$:
\begin{description}
\item{($r$-RR)}: $r\in \prules'$ and
\begin{eqnarray*}
\assign &=& (n,\pmodel',\pconst',\prules'\setminus \{r\},\porder'-[r],\pweight',\lweight',\derivation), \mbox{ with}\\
\derivation &=& (\drules'\setminus \{r\},\datoms',\dheads',\dbodies',\realhead').
\end{eqnarray*}

\item{($r$-RI)}:
\[
\assign \in \{(n,\pmodel',\pconst',\prules^*,\porder^*,\pweight',\lweight',\derivation) \mid
\porder^* \in (\porder' + [r])\},\mbox{ with}
\]
\[
\prules^* = \begin{cases}
  \prules' \cup \{r\}     & \mbox{if } \pconst' \models_{\innode{n}{\const}} r \\
  \prules' & \mbox{otherwise}
\end{cases}
\]
and one of the following two groups of properties has to be satisfied:

\begin{itemize}
\item ``$r$ is used'':
$\head \in \innode{n}{\const} \Ra (\head \in \dheads' \wedge \head > r)$,
for all $b \in \body \cap \innode{n}{\const}:$ $b \in \pconst' \wedge b < r$, and
\[
\derivation = (\drules'\cup \{r\},\datoms',\dheads',\dbodies'\cup (\body \cap \innode{n}{\const}),\realhead^*), \mbox{ with}
\]
\[
\realhead^*(c) = \begin{cases}
  1 & \mbox{if } c = \head \\
  \realhead'(c) & \mbox{otherwise}.
\end{cases}
\]

\item ``$r$ is not used'': $\derivation = \derivation'$.
\end{itemize}

\item ($a$-AR):
$a \in \pmodel' \Ra a \in \datoms'$ and
\begin{eqnarray*}
\assign &=& (n,\pmodel'\setminus \{a\},\pconst',\prules',\porder'-[a],\pweight',\lweight',\derivation), \mbox{ with}\\
\derivation &=& (\drules',\datoms'\setminus \{a\},\dheads',\dbodies',\realhead').
\end{eqnarray*}

\item ($a$-AI):
One of the following two groups of properties has to be satisfied:

\begin{itemize}
\item ``set $a$ to false'':
\[
\assign = (n,\pmodel',\pconst',\prules',\porder',\pweight^*,\lweight^*,\derivation'), \mbox{ with}\\
\]
$\pweight^*(c) = \pweight'(c) + \car(c,\pmodel',\innode{n}{\atoms}) - \car(c,\pmodel',\innode{n'}{\atoms})$, and\\
$\lweight^*(c) = \lweight'(c) + \car_<(c,\pmodel',\innode{n}{\atoms},\porder') - \car_<(c,\pmodel',\innode{n'}{\atoms},\porder')$.

\item ``set $a$ to true'':
\begin{eqnarray*}
\assign &\in&  \{(n,\pmodel^*=\pmodel'\cup \{a\},\pconst',\prules',\porder^*,\pweight^*,\lweight^*,\derivation) \mid\\
&& \porder^* \in (\porder' + [a])\}, \mbox{ with}\\
\derivation &=& (\drules',\datoms'\cup\datoms^*,\dheads',\dbodies',\realhead'), \mbox{ where}
\end{eqnarray*}
\[
\datoms^* = \begin{cases}
  \{a\} & \mbox{if } \exists c \in \dheads', a\in \clause(c), a>c\\
  \emptyset & \mbox{otherwise},
\end{cases}
\]
$\pweight^*(c) = \pweight'(c) + \car(c,\pmodel^*,\innode{n}{\atoms}) - \car(c,\pmodel',\innode{n'}{\atoms})$, and\\
$\lweight^*(c) = \lweight'(c) + \car_<(c,\pmodel^*,\innode{n}{\atoms},\porder^*) - \car_<(c,\pmodel',\innode{n'}{\atoms},\porder')$.
\end{itemize}

\item ($c$-CR):
$c \in \pconst' \LR \low(c) \leq \pweight'(c) \leq \upper(c)$,
$c \in \dheads' \Ra \realhead'(c) = 1$,
$c \in \dbodies' \Ra \lweight'(c) \geq \low(c)$, and
\begin{eqnarray*}
\assign&=&(n,\pmodel',\pconst'\setminus \{c\},\prules',\porder'-[c],\pweight',\lweight',\derivation), \mbox{ with}\\
\derivation&=&(\drules',\datoms',\dheads'\setminus \{c\},\dbodies'\setminus \{c\},\realhead').
\end{eqnarray*}

\item ($c$-CI):
One of the following two groups of properties has to be satisfied:

\begin{itemize}
\item ``set $c$ to false'':
$c \not\in \body \wedge c \neq \head$ for all $r \in \drules'$, and
\[
\assign=(n,\pmodel',\pconst',\prules'\cup\prules^*,\porder',\pweight'\cup\pweight^*,\lweight',\derivation'),\mbox{ with}
\]
$\prules^* = \{r \in \innode{n}{\rules} \mid \pconst' \models_{\innode{n}{\const}} r\}$, and
$\pweight^* = \{(c,\car(c,\pmodel',\innode{n}{\atoms})\}$.

\item ``set $c$ to true'':
$(c\in \body \Ra r > c) \wedge (c = \head \Ra r < c)$ for all $r \in \drules'$, and
\begin{eqnarray*}
\assign \hspace{-0.1cm} &\in& \hspace{-0.1cm}\{(n,\pmodel',\pconst^* = \pconst'\cup \{c\},\prules'\cup\prules^*,\porder^*,\pweight^*,\lweight^*,\derivation)\mid \\
&& \porder^*\in(\porder'+[c])\}, \mbox{ with}\\
\derivation \hspace{-0.1cm} &=& \hspace{-0.1cm} (\drules',\datoms'\cup \datoms^*,\dheads'\cup \dheads^*,\dbodies'\cup \dbodies^*,\realhead^*), \mbox{ where}
\end{eqnarray*}
$\prules^* = \{r \in \innode{n}{\rules} \mid \pconst^* \models_{\innode{n}{\const}} r\}$,
$\pweight^* = \pweight'\cup\{(c,\car(c,\pmodel',\innode{n}{\atoms})\}$,\\
$\lweight^* = \lweight'\cup\{(c,\car_<(c,\pmodel',\innode{n}{\atoms},\porder^*)\}$,
\begin{eqnarray*}
\dbodies^* &=& \begin{cases}
  \{c\} & \mbox{if } \exists r \in \drules': c \in \body\\
  \emptyset & \mbox{otherwise},
\end{cases}\\
\dheads^* &\in& \begin{cases}
  \{\{c\}\} & \mbox{if } \exists r \in \drules': c = \head\\
  \{\emptyset,\{c\}\} & \mbox{otherwise},
\end{cases}
\end{eqnarray*}
$\datoms^* = \{a \in \pmodel' \mid a \in \clause(c), c\in \dheads^*, a > c\}$, and\\
$\realhead^*(c) = 1 \LR c \in \dheads^* \wedge \exists r \in \drules: \head=c$.
\end{itemize}
\end{description}
\end{definition}

\smallskip
\noindent
For branch nodes, we
extend (with slight abuse of notation) $\prec_\TD$ to a ternary relation.

\begin{definition}
\label{def:branch}
Let
$\assign=(n,\pmodel,\pconst,\prules,\porder,\pweight,\lweight,\derivation)$,
$\assign'=(n',\pmodel',\pconst',\prules',\porder',\pweight',\lweight',\derivation')$, and
$\assign''=(n'',\pmodel'',\pconst'',\prules'',\porder'',\pweight'',\lweight'',\derivation'')$
be bag assignments for nodes $n, n', n'' \in T$ with
$\derivation = (\drules,\datoms,\dheads,\dbodies,\realhead)$,
$\derivation' = (\drules',\datoms',\dheads',\dbodies',\realhead')$, and
$\derivation'' = (\drules'',\datoms'',\dheads'',\dbodies'',\realhead'')$.
We relate $(\assign',\assign'')\prec_\TD \assign$ if
$n$ has two children $n'$ and $n''$ and the following conditions are fulfilled.
\begin{itemize}

\item $\pmodel = \pmodel' = \pmodel''$\quad
      $\pconst = \pconst' = \pconst''$
\item $\prules = \prules' \cup \prules''$\quad
      $\porder = \porder' = \porder''$
\item $\pweight(c) = \pweight'(c) + \pweight''(c) - \car(c,\pmodel,\innode{n}{\atoms})$ for all $c\in\innode{n}{\const}$
\item $\lweight(c) = \lweight'(c) + \lweight''(c) - \car_<(c,\pmodel,\innode{n}{\atoms},\porder)$ for all $c\in\pconst$
\item $\drules = \drules' = \drules''$\quad
      $\datoms = \datoms' \cup \datoms''$
\item $\dheads = \dheads' = \dheads''$\quad
      $\dbodies = \dbodies' \cup \dbodies''$
\item $\realhead(c) = \max\{\realhead'(c),\realhead''(c)\}$ for all $c\in\dheads$
\end{itemize}
\end{definition}

What follows is a small example which demonstrates how this $\prec_\TD$ relation is used to solve the consistency problem for PCCs.
Thereby we start with the only possible bag model $\assign=(n,\emptyset,\emptyset,\emptyset,[],\emptyset,\emptyset,\derivation)$ and $\derivation = (\emptyset,\emptyset,\emptyset,\emptyset,\emptyset)$ for each leaf node.
Now we traverse through the tree decomposition and calculate for each node all the bag assignments according to the relation $\prec_\TD$.
Finally, we check whether for the root node any such bag assignment could be generated.

\begin{example}
\label{ex:algo}
We are given a PCC $\Pi=(\{p_1,p_2\},\{c_1,c_2\},\{r_1\})$ with $c_1=(\{(p_1,1)\},1,1)$, $c_2 =(\{(\neg p_2,1)\},1,1)$, and $r_1 =(c_1,\{c_2\})$.

Its incidence graph as well as a normalized tree decomposition of width 1 are depicted in Figure~\ref{fig:ex}.
What follows is a list of all the bag assignments that can be computed according to the relation $\prec_\TD$, starting from the trivial bag assignments of the empty leaf nodes.

\smallskip
\noindent
Node $n_{1}$: (L)
\begin{align*}
  \vartheta_{1} &= (n_{1},\emptyset,\emptyset,\emptyset,[],\emptyset,\emptyset,(\emptyset,\emptyset,\emptyset,\emptyset,\emptyset))
\end{align*}

\noindent
Node $n_{2}$: ($p_1$-AI)
\begin{align*}
  \vartheta_{2,1} &= (n_{2},\emptyset,\emptyset,\emptyset,[],\emptyset,\emptyset,(\emptyset,\emptyset,\emptyset,\emptyset,\emptyset))\\
  \vartheta_{2,2} &= (n_{2},\{p_1\},\emptyset,\emptyset,[p_1],\emptyset,\emptyset,(\emptyset,\emptyset,\emptyset,\emptyset,\emptyset))
\end{align*}

\noindent
Node $n_{3}$: ($c_1$-CI)
\begin{align*}
  \vartheta_{3,1} &= (n_{3},\emptyset,\emptyset,\emptyset,[],\{(c_1,0)\},\emptyset,(\emptyset,\emptyset,\emptyset,\emptyset,\emptyset))\\
  \vartheta_{3,2} &= (n_{3},\emptyset,\{c_1\},\emptyset,[c_1],\{(c_1,0)\},\{(c_1,0)\},(\emptyset,\emptyset,\emptyset,\emptyset,\emptyset))\\
  \vartheta_{3,3} &= (n_{3},\emptyset,\{c_1\},\emptyset,[c_1],\{(c_1,0)\},\{(c_1,0)\},(\emptyset,\emptyset,\{c_1\},\emptyset,\{(c_1,0)\}))\\
  \vartheta_{3,4} &= (n_{3},\{p_1\},\emptyset,\emptyset,[p_1],\{(c_1,1)\},\emptyset,(\emptyset,\emptyset,\emptyset,\emptyset,\emptyset))\\
  \vartheta_{3,5} &= (n_{3},\{p_1\},\{c_1\},\emptyset,[c_1,p_1],\{(c_1,1)\},\{(c_1,0)\},(\emptyset,\emptyset,\emptyset,\emptyset,\emptyset))\\
  \vartheta_{3,6} &= (n_{3},\{p_1\},\{c_1\},\emptyset,[c_1,p_1],\{(c_1,1)\},\{(c_1,0)\},(\emptyset,\{p_1\},\{c_1\},\emptyset,\{(c_1,0)\}))\\
  \vartheta_{3,7} &= (n_{3},\{p_1\},\{c_1\},\emptyset,[p_1,c_1],\{(c_1,1)\},\{(c_1,1)\},(\emptyset,\emptyset,\emptyset,\emptyset,\emptyset))\\
  \vartheta_{3,8} &= (n_{3},\{p_1\},\{c_1\},\emptyset,[p_1,c_1],\{(c_1,1)\},\{(c_1,1)\},(\emptyset,\emptyset,\{c_1\},\emptyset,\{(c_1,0)\}))
\end{align*}

\noindent
Node $n_{4}$: ($p_1$-AR)
\begin{align*}
  \vartheta_{4,1} &= (n_{4},\emptyset,\emptyset,\emptyset,[],\{(c_1,0)\},\emptyset,(\emptyset,\emptyset,\emptyset,\emptyset,\emptyset))\\
  \vartheta_{4,2} &= (n_{4},\emptyset,\{c_1\},\emptyset,[c_1],\{(c_1,0)\},\{(c_1,0)\},(\emptyset,\emptyset,\emptyset,\emptyset,\emptyset))\\
  \vartheta_{4,3} &= (n_{4},\emptyset,\{c_1\},\emptyset,[c_1],\{(c_1,0)\},\{(c_1,0)\},(\emptyset,\emptyset,\{c_1\},\emptyset,\{(c_1,0)\}))\\
  \vartheta_{4,4} &= (n_{4},\emptyset,\{c_1\},\emptyset,[c_1],\{(c_1,1)\},\{(c_1,0)\},(\emptyset,\emptyset,\{c_1\},\emptyset,\{(c_1,0)\}))
\end{align*}

\noindent
Node $n_{5}$: ($r_1$-RI)
\begin{align*}
  \vartheta_{5,1} &= (n_{5},\emptyset,\emptyset,\emptyset,[r_1],\{(c_1,0)\},\emptyset,(\emptyset,\emptyset,\emptyset,\emptyset,\emptyset))\\
  \vartheta_{5,2} &= (n_{5},\emptyset,\{c_1\},\{r_1\},[r_1,c_1],\{(c_1,0)\},\{(c_1,0)\},(\emptyset,\emptyset,\emptyset,\emptyset,\emptyset))\\
  \vartheta_{5,3} &= (n_{5},\emptyset,\{c_1\},\{r_1\},[c_1,r_1],\{(c_1,0)\},\{(c_1,0)\},(\emptyset,\emptyset,\emptyset,\emptyset,\emptyset))\\
  \vartheta_{5,4} &= (n_{5},\emptyset,\{c_1\},\{r_1\},[r_1,c_1],\{(c_1,0)\},\{(c_1,0)\},(\{r_1\},\emptyset,\{c_1\},\emptyset,\{(c_1,1)\}))\\
  \vartheta_{5,5} &= (n_{5},\emptyset,\{c_1\},\{r_1\},[r_1,c_1],\{(c_1,0)\},\{(c_1,0)\},(\emptyset,\emptyset,\{c_1\},\emptyset,\{(c_1,0)\}))\\
  \vartheta_{5,6} &= (n_{5},\emptyset,\{c_1\},\{r_1\},[c_1,r_1],\{(c_1,0)\},\{(c_1,0)\},(\emptyset,\emptyset,\{c_1\},\emptyset,\{(c_1,0)\}))\\
  \vartheta_{5,7} &= (n_{5},\emptyset,\{c_1\},\{r_1\},[r_1,c_1],\{(c_1,1)\},\{(c_1,0)\},(\{r_1\},\emptyset,\{c_1\},\emptyset,\{(c_1,1)\}))\\
  \vartheta_{5,8} &= (n_{5},\emptyset,\{c_1\},\{r_1\},[r_1,c_1],\{(c_1,1)\},\{(c_1,0)\},(\emptyset,\emptyset,\{c_1\},\emptyset,\{(c_1,0)\}))\\
  \vartheta_{5,9} &= (n_{5},\emptyset,\{c_1\},\{r_1\},[c_1,r_1],\{(c_1,1)\},\{(c_1,0)\},(\emptyset,\emptyset,\{c_1\},\emptyset,\{(c_1,0)\}))
\end{align*}

\begin{figure}
\centering
\begin{tikzpicture}[node distance=1.5cm]
  \tikzstyle{every node}=[draw,circle,fill,inner sep=2pt]
  \node (c1) [label=left:$c_1$] {};
  \node (c2) [below of=c1,label=left:$c_2$] {};
  \node (p1) [right of=c1, label=right:$p_1$] {};
  \node (p2) [below of=p1,label=right:$p_2$] {};
  \node (r1) [label=left:$r_1$] at ++(-1,-0.75) {};
  \draw [-] (p1) -- (c1);
  \draw [-] (p2) -- (c2);
  \draw [-] (r1) -- (c1);
  \draw [-] (r1) -- (c2);
\end{tikzpicture}
\hspace{0.5cm}
\begin{tikzpicture}
\tikzstyle{every node}=[draw,rectangle]
\tikzstyle{level 1}=[level distance=0.75cm, sibling distance=2.5cm]
\node [label=left:$n_{14}$] {}
  child {node [label=left:$n_{13}$] {$r_1$}
    child {node [label=left:$n_{6}$] {$r_1$}
      child {node [label=left:$n_{5}$] {$c_1,r_1$}
        child {node [label=left:$n_{4}$] {$c_1$}
          child {node [label=left:$n_{3}$] {$c_1,p_1$}
            child {node [label=left:$n_{2}$] {$p_1$}
              child {node [label=left:$n_{1}$] {}}}}}}}
    child {node [label=left:$n_{12}$] {$r_1$}
      child {node [label=left:$n_{11}$] {$c_2,r_1$}
        child {node [label=left:$n_{10}$] {$c_2$}
          child {node [label=left:$n_{9}$] {$c_2,p_2$}
            child {node [label=left:$n_{8}$] {$p_2$}
              child {node [label=left:$n_{7}$] {}}}}}}}};
\end{tikzpicture}
\caption{Incidence graph (left) and tree decomposition (right) of Example~\ref{ex:algo}.}
\label{fig:ex}
\end{figure}
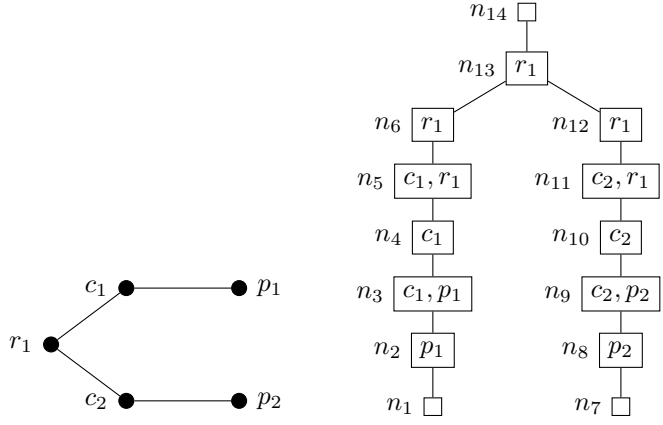

\noindent
Node $n_{6}$: ($c_1$-CR)
\begin{align*}
  \vartheta_{6,1} &= (n_{6},\emptyset,\emptyset,\emptyset,[r_1],\emptyset,\emptyset,(\emptyset,\emptyset,\emptyset,\emptyset,\emptyset))\\
  \vartheta_{6,2} &= (n_{6},\emptyset,\emptyset,\{r_1\},[r_1],\emptyset,\emptyset,(\{r_1\},\emptyset,\emptyset,\emptyset,\emptyset))
\end{align*}

The branch of nodes $n_7$ till $n_{12}$ is very similar to nodes $n_1$ till $n_6$.
Therefore we just present the bag assignments for $n_{12}$.

\smallskip
\noindent
Node $n_{12}$: ($c_2$-CR)
\begin{align*}
  \vartheta_{12,1} &= (n_{12},\emptyset,\emptyset,\emptyset,[r_1],\emptyset,\emptyset,(\emptyset,\emptyset,\emptyset,\emptyset,\emptyset))\\
  \vartheta_{12,2} &= (n_{12},\emptyset,\emptyset,\emptyset,[r_1],\emptyset,\emptyset,(\{r_1\},\emptyset,\emptyset,\emptyset,\emptyset))
\end{align*}

\noindent
Node $n_{13}$: (B)
\begin{align*}
  \vartheta_{13,1} &= (n_{13},\emptyset,\emptyset,\emptyset,[r_1],\emptyset,\emptyset,(\emptyset,\emptyset,\emptyset,\emptyset,\emptyset))\\
  \vartheta_{13,2} &= (n_{13},\emptyset,\emptyset,\{r_1\},[r_1],\emptyset,\emptyset,(\{r_1\},\emptyset,\emptyset,\emptyset,\emptyset))
\end{align*}

\noindent
Node $n_{14}$: ($r_1$-RR)
\begin{align*}
  \vartheta_{14} &= (n_{14},\emptyset,\emptyset,\emptyset,[],\emptyset,\emptyset,(\emptyset,\emptyset,\emptyset,\emptyset,\emptyset))
\end{align*}

Since $\vartheta_{14}$ could be derived, the example is a yes-instance of the consistency problem.
Indeed it has exactly one answer set $\{p_1\}$.
\exend
\end{example}

Let us look exemplarily at (CR) nodes in more detail.
Consider nodes $n$ which remove a constraint $c$, i.e., $\bag(n)=\bag(n') \setminus \{c\}$, where $n'$ is the child of $n$
(see, for instance,
the node with bag $\{p_3,c_3\}$ in the left branch of $\TD_{Ex}$ in
Figure~\ref{fig:td}, which is a $c_4$-removal node).
Let $\assign'=(n',\pmodel',\pconst',\prules',\porder',\pweight',\lweight',\derivation')$
with $\derivation' = (\drules',\datoms',\dheads',\dbodies',\realhead')$ be a bag model for $n'$.
We then create a bag model for $n$ as follows:
First we have to check whether the conditions
$c \in \pconst' \LR \low(c) \leq \pweight'(c) \leq \upper(c)$,
$c \in \dheads' \Ra \realhead'(c) = 1$, and
$c \in \dbodies' \Ra \lweight'(c) \geq \low(c)$ are satisfied.
Note that those checks correspond to the conditions 1, 10, and 8 of
Definition~\ref{def:p-solution}. They ensure that all
guesses with respect to $c$ are correct.
In the case of an affirmative answer,
we remove $c$ from all sets of $\assign'$ in order to create the new bag model
$\assign = (n,\pmodel',\pconst'\setminus \{c\},\prules',\porder'-[c],\pweight',\lweight',\derivation)$ with
$\derivation = (\drules',\datoms',\dheads'\setminus \{c\},\dbodies'\setminus \{c\},\realhead')$.

The following two theorems state that the rules defined above indeed
help in finding bag models.

\begin{theorem}[Soundness]
\label{the:soundness}
Given a bag model $\assign'$ (respectively bag models $\assign'$ and $\assign''$).
Then each bag assignment $\assign$
with $\assign' \prec_\TD \assign$ (respectively
$(\assign',\assign'') \prec_\TD \assign$) is a bag model.
\end{theorem}

\begin{proof}
Let $\assign'$ be a bag model for $n' \in \tree$ and let
$\assign$ be a bag assignment for node $n \in \tree$ with $\assign' \prec_\TD \assign$.
Then $n'$ is the single child of $n$, with $n$ being of type
(RR), (RI), (AR), (AI), (CR), or (CI). Assume $n$ is a ($r$-RR) node.
According to Definition~\ref{def:rules},
we have $r \in \prules'$ with $\assign$ and $\assign'$ differing only in $\prules = \prules' \setminus \{r\}$,
$\porder=\porder'-[r]$, and $\drules=\drules'\setminus\{r\}$.
Since $\assign'$ is a bag model, there exists a partial solution $\eassign$ of $n'$,
satisfying all the conditions of Definition~\ref{def:model}.

Claim: $\eassign$ is also a partial solution of $n$.

\noindent
To verify this claim, we have to check the conditions of
Definition~\ref{def:p-solution}.
Since $\insubtree{n'}{\const} = \insubtree{n}{\const}$,
$\intree{n'}{\const} = \intree{n}{\const}$,
$\insubtree{n'}{\atoms} = \insubtree{n}{\atoms}$,
$\intree{n'}{\atoms} = \intree{n}{\atoms}$, and
$\intree{n'}{\rules} = \intree{n}{\rules}$, the only non-trivial
condition is number~2 where we have to check
$\insubtree{n}{\rules} \subseteq \erules$.
Since $r \in \prules'$ and $\prules' = \erules \cap \innode{n'}{\rules}$,
we have $r \in \erules$.
Hence, from $\insubtree{n'}{\rules} \subseteq \erules$
follows that $\insubtree{n}{\rules} = \insubtree{n'}{\rules} \cup \{r\} \subseteq \erules$.

Furthermore, the projection of $\eassign$ to the bag $\bag(n)$ is exactly
$\assign$, since $\assign'$ and $\assign$ differ only by the fact,
that $r$ is removed from every set in $\assign$. Therefore $\assign$ is a bag model.
Analogously the theorem can be checked for the other five node types above.

Now let $\assign'$ and $\assign''$ be bag models for $n',n'' \in T$ and
let $\assign$ be a bag assignment for node $n \in \tree$ with $(\assign',\assign'') \prec_\TD \assign$.
Then $n$ has two children $n'$ and $n''$ and all the properties of Definition~\ref{def:branch}
are satisfied. Since $\assign'$ and $\assign''$ are bag models,
there exist partial solutions $\eassign'$ of $n'$ and $\eassign''$ of $n''$.
Using these two partial solutions we construct
$\eassign=(n,\emodel'\cup\emodel'',\econst'\cup\econst'',\erules'\cup\erules'',\eorder,\eweight,\elweight,\ederivation)$ with $\ederivation = (\edrules'\cup\edrules'',\edatoms'\cup\edatoms'',\edheads'\cup\edheads'',\edbodies'\cup\edbodies'',\erealhead)$. Thereby $\eorder \in (\eorder'+\eorder'')$,
\begin{align*}
\eweight(c)& = \begin{cases}
\eweight'(c) & c \in \insubtree{n'}{\const},\\
\eweight''(c) & c \in \insubtree{n''}{\const},\\
\eweight'(c) + \eweight''(c) - \car(c,\innode{n}{\emodel},\innode{n}{\atoms}) & \mbox{otherwise},
\end{cases}
\end{align*}
\begin{align*}
\elweight(c)& = \begin{cases}
\elweight'(c) & c \in \insubtree{n'}{\const},\\
\elweight''(c) & c \in \insubtree{n''}{\const},\\
\elweight'(c) + \elweight''(c) - \car_<(c,\innode{n}{\emodel},\innode{n}{\atoms},\eorder) & \mbox{otherwise},
\end{cases}\\
\erealhead(c)& = \begin{cases}
\erealhead'(c) & c \in \edheads'\setminus\edheads'',\\
\erealhead''(c) & c \in \edheads''\setminus\edheads',\\
\max\{\erealhead'(c),\erealhead''(c)\} & \mbox{otherwise}.
\end{cases}
\end{align*}
One can now check the conditions of Definition~\ref{def:p-solution} in order
to verify that $\eassign$ is a partial solution for $n$.
Furthermore, our construction ensures that the projection of $\eassign$ to the bag $\bag(n)$
is exactly $\assign$, which is therefore a bag model.
\end{proof}

\begin{theorem}[Completeness]
\label{the:completeness}
Given a bag model $\assign$ for node $n\in \tree$.
Then either $n$ is a leaf node, or there exists
a bag model $\assign'$ (respectively two bag models
$\assign'$ and $\assign''$) with
$\assign' \prec_\TD \assign$ (respectively
$(\assign',\assign'') \prec_\TD \assign$).
\end{theorem}

\begin{proof}
Again, we have to distinguish between the node type of $n$.
For instance, 
let $n \in \tree$ be an (r-RR) node with child $n'$, let
$\assign$ be a bag model for $n$.
We have to show that there exists a bag model $\assign'$ for $n'$
with $\assign' \prec_\TD \assign$.
Since $\assign$ is a bag model, there exists a partial solution $\eassign$ of $n$,
satisfying all the conditions of Definition~\ref{def:model}.
From $r \in \insubtree{n}{\rules}$ follows, that $r \in \erules$.
Now consider the projection of $\eassign$ onto the bag of $n'$.
Then the result is a bag model $\assign'$ of $n'$ satisfying the conditions
of Definition~\ref{def:model} and having $r \in \prules'$.
But then it is easy to check, that $\assign' \prec_\TD \assign$,
which closes the proof for (RR) nodes.
Analogously the theorem can be checked for the other six node types.
\end{proof}

Theorem~\ref{the:soundness} and Theorem~\ref{the:completeness}
show, that starting from the trivial bag models for empty leafs,
the dynamic programming algorithm creates all bag models for
the root node.
According to Proposition~\ref{prop:rootmodel},
those bag models are all we need to know.
Thus, this dynamic programming algorithm solves
the consistency problem.

\begin{theorem}
\label{the:non-uniform-ptime}
The consistency problem for PCCs $\prog$
can be solved in time
$\bigO(2^{6w} w! k^{4w} \cdot \size{\prog})$
with $w = \tw(\prog)$ and $k = \ww(\prog)$.
\end{theorem}

\begin{proof}
We first show that the number of different bag models at each node $n\in \tree$ is bounded.
The number of possible sets $\pmodel,\pconst,\prules$ is bounded by $2^w$, there are at most
$w!$ different orderings $\porder$, the number of cardinality functions $\pweight,\lweight$ is bounded by
$k^{2w}$, the number of possible sets $\drules,\dheads$ as well as $\datoms,\dbodies$ is bounded by $2^w$ each, and finally the number of check functions $\realhead$ is bounded by $2^w$.
This leads to at most $2^{4w} w! k^{2w}$ many different bag models at node $n$.
At each node the effort to compute a single bag model is constant
with the exception of branch nodes, where one has to compare possible pairs
of bag models of each child node.
Thereby only pairs are combined which have identical $\pmodel,\pconst,\prules,\porder,\drules,\dheads$.
This means for each bag model of the first child node there are at most $2^{2w} k^{2w}$
(the number of possible functions/sets $\pweight,\lweight,\datoms,\dbodies,\realhead$) bag models at the second child to consider.
The time per node is therefore bounded by $2^{6w} w! k^{4w}$ and since the number of nodes in our tree decomposition is bounded by $\bigO(\size{\prog})$, the total time of
$\bigO(2^{6w} w! k^{4w} \cdot \size{\prog})$ follows.
\end{proof}

\section{Extensions}
\label{sect:extensions}

In this section, we discuss some extensions of our dynamic
programming approach and of Theorem~\ref{the:non-uniform-ptime}.

\paragraph{PWCs with unary weights.}
Our dynamic programming algorithm for the consistency problem of PCCs can be easily extended
to PWCs with \emph{unary representation} both, of the weights and of the 
constraint bounds (\emph{PWCs with unary weights}, for short).

\begin{theorem}
\label{the:unaryPWCs-non-uniform-ptime}
Given an arbitrary PWC $\prog$.
The consistency problem for PWCs
with unary weights
can be solved in time
$\bigO(2^{6w} w! k^{4w} \cdot \size{\prog})$
with $w = \max(3,\tw(\prog))$ and $k = \ww(\prog)$.
\end{theorem}

\begin{proof}
It suffices to show that
every PWC $\prog$ with unary weights
can be efficiently transformed into a PCC $\prog'$ such that
$\prog$ is only linearly bigger than
$\prog$, the constraint-width remains the same, and the treewidth
is  $\max(3,\tw(\prog))$.
The transformation from $\prog$ to
$\prog'$ processes each literal $\ell$ with weight $j > 1$ in each constraint $c$ of $\prog$ as follows: reduce the weight of $\ell$ to 1 and
add $j-1$ fresh atoms
$\ell_2, \ldots, \ell_j$
(each of weight 1) to $c$. Moreover, we add,
for $\alpha \in \{2, \ldots, j\}$,
new constraints
$c_\alpha :=  (\{ (\ell,1), (\neg \ell_\alpha,1) \}, 1, 1)$
and new rules
$r_\alpha :=  ( c_\alpha, \emptyset)$
to ensure that the fresh variables $\ell_2, \ldots, \ell_j$ have the same truth
value as $\ell$ in every model of $\prog$.

\smallskip
\noindent
It is easy to check that $\prog'$ is only linearly bigger than $\prog$
(since $j$ is given in unary representation) and that the
constraint-width and treewidth are not increased (resp.\ changed from
treewidth $\leq 2$ to  treewidth~3).
\end{proof}

\paragraph{Reasoning with PCCs and PWCs with unary weights.}
In non-monotonic reasoning, two kinds of reasoning are usually
considered, namely skeptical and credulous reasoning.
Recall
that an atom $a$ is skeptically 
implied by a program $\Pi$ if $a$ is true
(i.e., contained)
in every 
stable model of $\Pi$. 
Likewise, an atom $a$ is credulously implied by 
$\Pi$ if $a$ is true in some stable model of $\Pi$.
Our algorithm for the consistency problem can be easily extended to an algorithm for skeptical or
credulous reasoning with PCCs and PWCs with unary weights. The above upper bounds on the complexity thus carry over from the consistency problem to the reasoning problems. We only work out the PCC-case below:

\begin{theorem}
\label{the:reasoning-with-PCCs}
Both the skeptical and the credulous reasoning problem for PCCs $\prog$
can be solved in time
$\bigO(2^{6w} w! k^{4w} \cdot \size{\prog})$
with $w = \tw(\prog)$ and $k = \ww(\prog)$.
\end{theorem}

\begin{proof}
Suppose that we are given a PCC $\prog$
and an atom $a$.
The dynamic programming algorithm for the
consistency problem has to be extended
in such a way that we additionally maintain
two flags $\crr(\assign)$ and $\skr(\assign)$
for every bag assignment $\assign$. These flags
may take one of the values $\{\bot,\top\}$ with the intended meaning that
$\crr(\assign) = \top$
(resp.\ $\skr(\assign) = \top$) if and only if there exists
a partial solution $\eassign=(n,\emodel, \dots)$,
(resp.\ if and only if for all
partial solutions $\eassign=(n,\emodel, \dots)$)
the atom $a$ is true in $\emodel$. Otherwise this flag is set to $\bot$.
Then $a$ is credulously (resp.\ skeptically) implied by $\prog$ if and only if there exists a bag
model   (resp.\ if and only if for all bag models)  $\assign$ of the root node $\rootnode$ of $\tree$, we have \ $\crr(\assign) = \top$
(resp.\ $\skr(\assign) = \top$).
Clearly, maintaining the two flags fits within the desired complexity bound.
\end{proof}

\paragraph{Bounded treewidth and bounded constraint-width.}
Recall that we have proved the fixed-parameter linearity of the
consistency problem of PWCs when treewidth and constraint-width are
taken as parameter (see Theorem~\ref{the:lintime}). This fixed-parameter
linearity result (as well as the analogous result for the skeptical and
credulous reasoning problem which can be easily seen to be expressible
in MSO logic) could also be obtained as a corollary of
Theorem~\ref{the:unaryPWCs-non-uniform-ptime}.  Indeed, consider a PWC
$\Pi$ whose treewidth $w$ and constraint-width $k$ are bounded by some
fixed constant.  By previous considerations, we may thus assume that all
weights occurring in $\Pi$ are bounded by a constant. Therefore, we can
transform all weights and bounds into unary representation such that the
size of the resulting PWC with unary weights differs from $\size{\prog}$
only by a constant factor (namely $2^k$). The upper bound on the
complexity in Theorem~\ref{the:unaryPWCs-non-uniform-ptime} immediately
yields the desired fixed-parameter linearity result since $f(w)\cdot
\bigO(k^{2w})$ is bounded by a constant that is independent of the size
of $\prog$.

\section{W[1]-Hardness}
\label{sect:w1}

In this section we will show that it is unlikely that one can improve the
non-uniform polynomial-time result of Theorem~\ref{the:non-uniform-ptime} to a
fixed-parameter tractability result (without bounding the con\-straint-width as in
Theorem~\ref{the:lintime}).  We will develop our hardness result within the
framework of \emph{parameterized complexity}. Therefore we first outline some
of the main concepts of the subject, for an in-depth treatment we refer to
other sources \cite{DowneyFellows99,FlumGrohe06,Niedermeier06}.

An instance of a parameterized problem is a pair $(x,k)$, where $x$ is the
main part and~$k$ (usually a non-negative integer) is the parameter. A
parameterized problem is \emph{fixed-parameter tractable} if an instance
$(x,k)$ of size $n$ can be solved in time $O(f(k)n^c)$ where $f$ is a
computable function and $c$ is a constant independent of~$k$.  If $c=1$ then
we speak of linear-time fixed-parameter tractability. $\FPT$ denotes the class
of all fixed-parameter tractable decision problems.  Parameterized complexity
theory offers a completeness theory similar to the theory of NP-completeness.
An \emph{fpt-reduction} from a parameterized decision problem~$P$ to a
parameterized decision problem~$Q$ is a transformation that maps an
instance~$(x,k)$ of~$P$ of size $n$ to an instance~%
$(x',k')$ 
of~$Q$ 
with $k'\leq g(k)$
in
time~$O(f(k) n^{c})$ ($f,g$ are arbitrary computable functions, $c$ is a
constant) such that $(x,k)$ is a yes-instance of~$P$ if and only if
$(x',k')$ is a yes-instance of~$Q$.  A~parameterized complexity class $\cC$ is the
class of parameterized decision problems fpt-reducible to a certain
parameterized decision problem $Q$.
A parameterized problem $P$ is $\cC$\hy hard, if every problem in $\cC$ is fpt-reducible to $P$.
Problem $P$ is called $\cC$\hy complete, if it is additionally contained in $\cC$.
Of particular interest is the class
$\W{1}$ which is considered as the parameterized analog to {NP}.  For example,
the \textsc{Clique} problem (given a graph $G$ and an integer $k$, decide 
whether
$G$ contains a complete subgraph on $k$ vertices), parameterized by $k$, is a
well-known $\W{1}$\hy complete problem.  It is believed that $\FPT\neq \W{1}$,
and there is strong theoretical evidence that supports this belief, for
example, $\FPT=\W{1}$ would imply that the Exponential Time Hypothesis fails,
see~\cite{FlumGrohe06}.

In the proof of Theorem~\ref{the:whard} below we will devise an fpt-reduction
from the \textsc{Minimum Maximum Outdegree} problem (or MMO, for
short). To state this problem we need to introduce some concepts.  A~(positive
integral) \emph{edge weighting} of a graph $H=(V,E)$ is a mapping $w$ that
assigns to each edge of $H$ a positive integer.  An \emph{orientation} of $H$
is a mapping $\Lambda:E\rightarrow V\times V$ with $\Lambda(\{u,v\})\in
\{(u,v),(v,u)\}$. The \emph{weighted outdegree} of a vertex $v\in V$ with
respect to an edge weighting $w$ and an orientation $\Lambda$ is defined as
\[
d^+_{H,w,\Lambda}(v)=\sum_{\{v,u\}\in E \text{~with } \Lambda(\{v,u\})=(v,u)}
w(\{v,u\}).
\]
An instance of MMO consists of a graph $H$, an edge weighting $w$ of $H$, and
a positive integer~$r$; the question is whether there exists an orientation
$\Lambda$ of $H$ such that $d^+_{H,w,\Lambda}(v)\leq r$ for each $v\in V$.
The MMO problem with edge weights (and therefore also $r$) given in unary is $\W{1}$-hard when parameterized by the treewidth of~$H$~\cite{Szeider2010}.

\begin{theorem}\label{the:whard}
  The consistency problem for PCCs is $\W{1}$-hard when parameterized by
  treewidth.
\end{theorem}
\begin{proof}
  Let $(H,w,r)$ be an instance of MMO of treewidth~$t$, $H=(V,E)$.  We may
  assume that no edge is of weight larger than $r$ since otherwise we can
  reject the instance.  Let $\prec$ be an arbitrary linear ordering of $V$. We
  form a PWC $\prog=(\atoms,\const,\rules)$ with unary weights as
  follows: The set $\atoms$ contains an atom $a_{uv}=a_{vu}$ for each edge
  $\{u,v\}\in E$; $\const$ contains a constraint $c_v=(S_v,0,r)$ for each vertex
  $v\in V$ where $S_v=\SB (a_{uv},w(\{v,u\})) \SM \{u,v\}\in E,\ v \prec u\SE \cup \SB
  (\neg a_{uv},w(\{v,u\})) \SM \{u,v\}\in E,\ u \prec v \SE$; $\rules$ contains a rule
  $r_v=(c_v,\emptyset)$ for each vertex $v\in V$.

  \emph{Claim 1. $\tw(\prog)\leq \max(2,t)$}.  Let $(T,\chi)$ be a tree
  decomposition of $H$ of width $t$.  We extend $(T,\chi)$ to a tree
  decomposition of $\prog$ as follows.  For each edge $\{u,v\}\in E$  we pick a
  node $n_{uv}$ of $T$ with $u,v\in \chi(n_{uv})$ and for each vertex $v\in V$
  we pick a node $n_v$ of $T$ with $v\in \chi(n_v)$ (such nodes exist by the
  definition of a tree decomposition).  We attach to $n_{uv}$ a new neighbor
  $n_{uv}'$ (of degree 1) and put $\chi(n_{uv}')=\{u,v,a_{uv}\}$, and we
  attach to $n_v$ a new neighbor $n_v'$ (of degree 1) and put
  $\chi(n_v')=\{v,r_v\}$. It is easy to verify that we obtain this way a tree
  decomposition of  $\prog$ of width $\max(t,2)$, hence the claim
  follows.      Note that in fact we have $\tw(\prog)\geq \tw(H)$ since
  $H$ is a graph minor of the incidence graph of~$\prog$.

  \emph{Claim 2. $H$ has an orientation $\Lambda$ with $\max_{v\in V} d^+_{H,w,\Lambda}(v)\leq r$
  if and only if $\prog$ has a model}.  We associate with an orientation $\Lambda$ the subset
  $A_\Lambda=\SB a_{uv}\in A_\Lambda \SM u\prec v$ and $\Lambda(\{u,v\})=(u,v)\SE$.
  This gives a natural one-to-one correspondence between orientations of~$H$
  and subsets of $\atoms$.  We observe that for each $v\in V$, the sum of
  weights of the literals in constraint $c_v$ satisfied by $A_\Lambda$ is
  exactly the weighted outdegree of $v$ with respect to $\Lambda$. Hence
  $A_\Lambda$ is a model of $\prog$ if and only if $d^+_{H,w,\Lambda}(v)\leq r$ for all
  $v\in V$.

  \emph{Claim 3. All models of $\prog$ are stable}.  This claim follows by
  exactly the same argument as in the proof of Theorem~\ref{the:nphard}.

  $\prog$ can certainly be obtained from $(H,w,r)$ in polynomial
  time. We can even encode the weights of literals in unary since we
  assumed that that the edge weighting $w$ is given in unary.  Hence, by
  Claims~1--3 we have an fpt-reduction from MMO to the consistency
  problem for PWCs with unary weights. Using the construction as
  described in the proof of
  Theorem~\ref{the:unaryPWCs-non-uniform-ptime}, we can transform
  $\prog$ in polynomial time into a decision-equivalent PCC $\Pi'$ by
  increasing the treewidth at most by a small constant.  In total we
  have an fpt-reduction from MMO to the consistency problem for PCCs
  (both problems parameterized by treewidth).  The theorem now follows
  by the $\W{1}$-hardness of MMO for parameter treewidth.
\end{proof}

\section{Discussion}
\label{sect:discussion}

In this work, we have proved several results for PWCs and PCCs of bounded treewidth without addressing the problem of actually computing a tree decomposition of appropriate width. As has been mentioned earlier,
\cite{Bod96}
showed that
deciding if a graph has treewidth $\leq w$ and, if this is the case, computing a tree decomposition of width $w$ is
fixed-parameter linear %
for parameter $w$.
Unfortunately, this linear time algorithm is only of theoretical interest
and the practical usefulness is limited \cite{DBLP:journals/endm/KosterBH01}.
However, considerable progress has been recently made in developing  heuristic-based tree decomposition algorithms
which can handle graphs with moderate size of several hundreds of vertices
\cite{DBLP:journals/endm/KosterBH01,DBLP:journals/dm/BodlaenderK06,DBLP:journals/algorithmica/EijkhofBK07,DBLP:journals/cj/BodlaenderK08,Kask11}.
Recently a meta-theorem for MSO problems on graphs with
cardinality and weight constraints was shown~\cite{Szeider2011}.
This meta-theorem allows one to
handle cardinality constraints with respect to sets that occur as free
variables in the corresponding MSO formula. It provides a polynomial time
algorithm for checking whether a PCC (or a PWC with weights in unary) of
bounded treewidth has a model.  However, in order to check whether a PCC has a
\emph{stable} model, one needs to handle cardinality constraints with respect
to sets that occur as quantified variables in the MSO formula, which is not
possible with the above mentioned meta-theorem.

We have already mentioned  a
dynamic programming algorithm for ASP
\cite{JaklPW09}. This algorithm works
for programs without cardinality or weight constraints, but
possibly with disjunction in the head of the rules.
The data structure manipulated
at each node for this ASP algorithm is conceptually much simpler than the one used here: Potential models of the given program are represented by so-called tree-models. A tree-model consists of a subset of the atoms in a bag (the ones which are true in the models thus represented) and a subset of the rules in a bag (the ones which are
validated by the models thus represented). However, to handle
the minimality condition on stable models, it is not sufficient to
propagate potential models along the bottom-up traversal of the tree decomposition. In addition, it is required, for each potential model $M$, to keep track of all those models of the reduct w.r.t.\ $M$ which would prevent $M$ from  being minimal. Those models
are represented by a set of tree-models
accompanying each tree-model.
Hence, despite the simplicity of the data structure, the time
complexity of the algorithm from \cite{JaklPW09} is
{\em double exponential\/} in the treewidth, since it
has to handle {\em sets of subsets of the bag\/} at each node.
Therefore, rather than extending
that algorithm by mechanisms to handle weight or
cardinality constraints, we have presented here an algorithm based on a completely different data structure -- in particular, keeping track of orderings of the atoms. We have thus managed to obtain
an algorithm whose
time complexity
is {\em single exponential\/} in the treewidth.

\section{Conclusion}
\label{sect:conclusion}

In this paper we have shown how the notion of bounded treewidth can be used to identify tractable fragments of answer-set programming with weight constraints. However,
by proving hardness results,
we have also shown that a straightforward application of
treewidth is not sufficient to achieve the desired tractability.

The upper bounds on the time complexity of our
dynamic programming algorithms
were obtained by very coarse estimates
(see Theorems~\ref{the:non-uniform-ptime},
 \ref{the:unaryPWCs-non-uniform-ptime},
 \ref{the:reasoning-with-PCCs}).
In particular, we assumed
straightforward
methods for storing and manipulating bag assignments.
For an actual implementation of our algorithm, 
we plan to use the 
SHARP framework\footnote{\url{http://http://www.dbai.tuwien.ac.at/proj/sharp/}},
a C++ interface that enables rapid development of algorithms which are based on tree or hypertree decompositions by 
providing (hyper-)tree decomposition routines and algorithm interfaces. It thus allows
the designer to focus on the problem-specific part of the algorithm. 
SHARP itself uses the htdecomp library\footnote{\url{%
http://www.dbai.tuwien.ac.at/proj/hypertree/downloads.html}}
which implements several heuristics for (hyper)tree decompositions, see also \cite{Dermaku08}.
Using sophisticated methods and data structures 
in implementing the functionality of the different node types of our algorithm 
should eventually result in 
a further improvement of the (theoretical) upper bounds on the time complexity provided
in this paper.

For future work, we plan to extend the parameterized complexity analysis and the development of efficient algorithms to further problems where weights or cardinalities play a role. Note that weights are a common feature in
the area of knowledge representation and reasoning, %
for instance, to express costs or probabilities.

\bibliographystyle{abbrv}
\bibliography{publications}

\end{document}